\documentclass[12pt]{article}
\usepackage{caption}
\usepackage{graphicx}
\usepackage{amssymb} 
\usepackage{amsmath}    
\usepackage{latexsym}   
\usepackage{bm}
\usepackage{times}
\usepackage{centernot}
\usepackage{float}
\usepackage{color}
\usepackage{wrapfig}
\usepackage{mathtools} 
\usepackage{hyperref}
\usepackage{amsthm}
\usepackage{cleveref}
\usepackage{lineno}
\usepackage{float}

\newtheorem{prop}{Proposition}

\theoremstyle{definition}
\newtheorem{defn}{Definition}
\begin{document}
\title{Synchronizing rhythms of logic}
\author{John M. Myers$^{Ia}$ \vspace*{10pt}and  Hadi 
  Madjid$^{IIb}$ \\{\it 960 Waltham St. Apt. 168, Lexington, MA 02421$^a$,}\\
  {\it 309 Winthrop Terrace, Bedford \vspace*{10pt}MA 01730$^b$.}\\
  jmartmyers@gmail.com$^I$, \vspace*{-.1 in}gailmadjid@comcast.net$^{II}$}

\date{\today}
\maketitle

\begin{abstract} 
  Although quantum states nicely explain experiments, the outcomes of experiments are not states. Instead, outcomes correspond to probability distributions. Twenty years ago we proved categorically that probability distributions leave open a choice of quantum states to explain experiments  that is resolvable only by a move beyond logic, which, inspired or not, can be characterized as a guess.  Guesses link the inner lives of investigators to their explanations of experimental results.  Recognizing the inescapability of guesswork in physics leads to avenues of investigation, one  of which is presented here.

We invert the quest for the logical foundations of physics to reveal a physical
basis for logic and calculation, and we represent this basis
mathematically, in such a way as to show the shaping and re-shaping of
calculations by guesswork.  We draw on the interplay between guessing and
computation in digital contexts that, perhaps surprisingly, include living organisms.
Digital computation and communication depend on a type of synchronization that
coordinates transitions among physically distinct conditions represented by
``digits.'' This {\it logical synchronization}, known to engineers but neglected
in physics, requires guesswork for its maintenance.  By abstracting digital
hardware, we model the structure of human thinking as logically synchronized
computation, punctuated by guesses.

We adapt marked graphs to mathematically represent computation and represent guesses by unpredictable changes in these marked graphs. The marked graphs
reveal a logical substructure to spatial and temporal navigation, with
implications across physics and its biological applications.  By limiting our
model to the logical aspect of communications and computations---leaving out
energy, weight, shape, etc.---we unveil logical structure in relation to
guesswork, applicable not just to electronics but also to the functioning of
living organisms.
\end{abstract}

\newpage
\setcounter{tocdepth}{2}
\tableofcontents
\newpage

\section{Introduction}\label{sec:one}
One could always suppose  that guesswork is essential to physics; however a proof
provides a firm foundation on which to build.  Twenty years ago we proved that
guesswork is necessary to bridge logical gaps between evidence and any quantum
explanation of that evidence.  Acknowledging guesswork leads physics away from a
quest for final answers to a form of dialogue that expects and accepts
surprises.  The broad consequences of this include the following.
\begin{enumerate}
\item Choosing an explanation is an imaginative act, subject to the need for revision, so ``no final answers'' are to be sought.
\item Because of their dependence on guesswork, scientific explanations, although testable, have something of the imaginative character of metaphors.
\item The inner lives of physicists---unavailable to exterior, ``objective''
  view---participate in and influence the world that physicists measure with their
  clocks.
\item The source of any guess on which logical deductions depend is logically inexplicable. Andr\'e Malraux, Minister of
Culture in France from 1958 to 1969, interviewed leading artists about their
sources of inspiration.  Their answers led Malraux to speak of their inspiration
as coming from contact with what he called ``the Unknowable''
\cite[p.\ 98]{mask}. 
\end{enumerate}
So what might be the next step in appreciating ``no final answers''?
In this paper we develop the following three thoughts and some of their \vspace*{8pt}implications:\\
\setlength{\fboxrule}{1.2pt}
\fbox{\begin{minipage}{36em} {\bf (1) The Unknowable contributes to the unpredictable evolution of mathematics.\\ (2) Not just people but all living organisms depend on both computation and guesses.\\ (3) That dependence can be cartooned  mathematically.}
\end{minipage}} \vspace*{10pt}

We propose that logic and computation depend on how matter can behave in ``lumps'', i.e., exhibiting the distinct physical conditions, whether in patterns of  pebbles or the nucleotides strung like beads of a necklace in molecules of DNA.  
The distinct conditions must undergo transitions to other distinct conditions, as in the motion of beads on an abacus or the electronicic motion of digital-computing hardware.  The transitions among distinct conditions depend on a special
type of synchronization that differs
from the synchronization made famous in physics by special relativity. As shown below, logical synchronization requires the steering of clock rates,  and that steering requires guesswork.  Computation requires communications that must be logically synchronized and modern communications require computation.  As realized in  hardware, digital computation and communication are inseparably woven into logically synchronized {\it computational networks}.

Appreciation of the physical basis for computational networks, centered on logical synchronization with its need for guesswork, has implications that straddle  physics, biology, and mathematics.
\begin{enumerate}
\item In physics, we show how communication works as a topological
  foundation of space and time, involving guesswork in the maintenance of
  necessary logical synchronization.
\item In biology, avoiding any discussion of consciousness,
we model basal cognition along with human thinking as computation punctuated by unpredictable reprogramming.  
\item In discrete mathematics, we define marked graphs that are adapted to express the communications and computations in both digital systems and living organisms. Marked graphs are the only mathematical objects we have found that deal with motion without assuming a clock (i.e.,  a physical device), something necessarily outside of mathematics.  As shown below, they allow the definition of {\it logical distance} as a precursor to physical time and distance.
\end{enumerate}

\subsection{Background} In 2005, pursuing a vision experienced by HM in 1985, we proved that evidence, while constraining explanations, necessarily leaves open
an infinite range of conflicting explanations \cite{05aop}.  Thus, choosing an
explanation requires an extra-logical act. I (JMM) call such an act a {\it
  guess}.  What do I mean by ``guess''? As an individual scientist,  diverse thoughts flit into my mind when I am
puzzled. Some are fleeting, but there are others I pick up
and use as starting points for doing something, and these I call guesses.  For
me, guessing is contrasted with calculating.  The logical gap between evidence and
its explanations, calling for guesswork to link them, reveals an
unpredictability far more drastic than quantum uncertainty, and beyond any
fuzziness due to limited sample sizes.

Although it has important consequences, the proof of this has been largely
overlooked.  For one thing, verifying the proof requires knowledge omitted in
introductory courses on quantum mechanics.  In these courses, one learns how a
quantum state and a measurement operator on a vector space imply certain
probabilities of outcomes, but not about the so-called an {\it inverse problem},
the problem of using quantum theory to {\it explain} measured data.  The inverse
problem is important when something surprising is found, as in the
Stern--Gerlach experiment that led to half-integral spin \cite{SternGerlach}.

In a limited way, the inverse problem is addressed in quantum decision theory
\cite{helstrom,holevo} and in quantum information theory
\cite{nielsen}. However, one usually assumes a finite dimensional
vector space, and that the measurement operators are known (leading to
what is called {\it quantum tomography}).  However, both the operators and the
dimension are invisible to experiment,  so the assumption is logically
indefensible.  Thus, solving an inverse problem requires determining a
quantum state {\it and} measurement operators from the probabilities of
outcomes abstracted from experiments {\it without} assuming the dimension
of a vector space.  We proved that this problem cannot have a unique
solution, and indeed that infinitely many conflicting explanations fit whatever
evidence is at hand, so choosing among them is impossible without a guess. In
an application, we produced a second explanation of probabilities involved in a
quantum cryptography protocol, revealing an unsuspected security vulnerability
\cite{07tyler,Conditional}.

\subsection{The obstacle of global time}
The appreciation of logical synchronization is impeded by the widespread tacit assumption of ``universal time,'' as if the production of
time and space coordinates for scientific and other purposes were irrelevant to
theoretical physics.\footnote{``An important task of the theoretical physicist
lies in distinguishing between trivial and nontrivial discrepancies between
theory and experiment'' \cite[p.\ 3]{feshbach}.} 
Setting aside this assumption allows for the consideration of synchronization in digital
communications, which is essential to time distribution, computation, and indeed,
all logical operations.

To begin with, we notice how national and international time broadcasts depend
on digital communication networks that {\it construct time} from local clocks
whose rates must be adjusted.  These adjustments of the clocks {\it make}
time depend on guesswork.  For example, I set my clock using my phone, but how does my phone
tell the time?  It gets the time from the Internet, but where does the
Internet get it?  In the US, it gets it from Boulder, Colorado, where the
National Institute of Standards and Technology (NIST) has clocks. While my clock
tells the time, the clocks at NIST {\it make} time for the United States.  NIST
transmits readings of its clocks over the Internet, so that my phone can tell
time.  Why do I say ``clocks'' at NIST and not just ``clock''?  Like any
machines, the clocks at NIST need maintenance, and sometimes break and are
replaced, so NIST must use several clocks.  The clocks at NIST do not agree  exactly
with each other. NIST adjusts the clock rates to limit their deviations from each
other, and these adjustments require guesswork.  NIST and other national
metrology institutes maintain global networks of clocks, some on the ground
others on orbiting satellites, all undergoing rate adjustments, tuned by
guesswork \cite{tfr}.
\begin{quote}
The fact is that time as we now generate it is dependent
upon defined origins, a defined resonance in the cesium atom, interrogating
electronics, induced biases, time scale algorithms, and random perturbations
from the ideal. Hence, at a significant level, time---as man generates it by the
best means available to him---is an artifact. Corollaries to this are that every
clock disagrees with every other clock essentially always, and no clock keeps
ideal or ``true'' time in an abstract sense except as we may choose to define
it \cite{allan87}.  
\end{quote}

\subsection{The inseparability of computation and communication}
The digital communication and computation that pervade today's science depend on:
 (1) physically distinct conditions to distinguish between 0 and 1, (2)
 transitions between these conditions, (3) local timing ordered by local clocks that
 are synchronized differently from the usual manner in physics, and (4) guesswork.
 How these work together in computer networks is one of the main topics of this report.

Computations involve the communication of inputs and outputs over networks,
including the Internet.  In the other direction, digital communication requires
message processing, e.g. copying, addressing, searching, coding, and decoding, all
of which are instances of computation.  We elevate this back-and-forth
dependence to the following principle of wide application.
\begin{prop}
  Computation and communication are inseparable.
\end{prop}
We will speak of neither separately, but instead bundle them together, speaking
of {\it computational networks}.  These are found not only in electronic
computers but across all life.  In \cref{sec:2} we define computational
networks embedded in cyclic environments that produce inputs and consume
outputs, as in process control.  Computational networks express two levels of
unpredictability: (1) the inputs supplied by an environment are viewed as
unpredictable, and (2) unforeseen happenings change one computational network
into another. 

Computational networks perform logical operations that work in coordinated
 rhythms.  Without communication-guided adjustments, the rhythms of any two
 logical operations drift, as do any two clocks, so that phase deviations
 between them increase indefinitely.  To limit phase deviations, computational
 networks depend on a type of synchronization that differs from that which
 Einstein made famous in special relativity.  This {\it logical synchronization}
maintains the necessary  phase relations in the rhythms of logical operations.  To
 picture logical synchronization, think of playing a game of catch, tossing a
 ball back and forth with another person: you cycle through phases of tossing
 and catching, and if you are looking the wrong way you don't catch the ball.  
 Coordinating the phase ``ready to catch'' with the arrival of the tossed
 ball exemplifies logical synchronization.  As tangible examples of
 computational networks, we introduce ``token games'' in which you can participate
 by moving tokens over a game board.

 Although some contributions to phase deviations can be predicted, the
 statistical properties are non-stationary and unpredictable for others.  To
 maintain logical synchronization, rate adjustments guided by guesswork are
 necessary.\footnote{For a pendulum clock, one adjusts the tick rate by changing
 the length of the pendulum by a small amount. Typically, a screw at the bottom
 of the pendulum is turned to slightly raise or lower its center of gravity.}
 Phase deviations in logical synchronization limit the rate of information
 processing. The tighter the logical synchronization, the smaller the phase
 deviations, and the more information can flow.  Controlling drift in clocks
 is analogous to steering a car in a gusty cross wind: better control requires
 better guesses.
 
Here, we discuss only the logical aspect of computational networks; i.e., we say
nothing about energy, weight, shape, visual or auditory output, etc.  We discuss
the physical mechanisms of electronic computation only to show their ``lumpy''
character, leaving open whether the lumps are small or large and whether they are
lumps in frequencies of oscillations or in spatial extent or some other
modality.  Because we focus exclusively on the logical aspect, our formulation
of computational networks and their logical synchronization has applications not
only to man-made systems but also to biology, including human thought.
Rejecting the assumption that ``physical'' means ``entirely explicable,'' we
represent human thought as computing interspersed with unforeseeable reprogrammings.

In \cref{sec:3}, we represent the logic of computational networks by means of
graphs with markings comprised of tokens placed on arrows.  Computational steps are represented by changes in these markings.  Markings  of interest are partitioned into equivalence classes.  Following \cite{71MarkedGraph}, we call equivalence classes of markings {\it families}, and these play an important role in the following sections. The concept of {\it logical distance} is introduced and shown to depend not only  on the graph but also  on its {\it family}, of which there may be more than one.  Unpredictable changes in computation are expressed by unpredictable changes from one graph to another or by a change of the family for the graph.
 
As discussed in \cref{sec:4}, logical synchronization and its dependence on
extra-logical adjustment point to an alternative concept of time and space.  In
\cref{sec:4.2}, we show how measurements of computational networks that attempt
to determine a temporal order of events sometimes lead to nonsense.  An
alternative ``time'' based on logical synchronization, rather than special relativity, advances the identification of ``time'' with
communication that Einstein began in 1905.

\Cref{sec:5} pulls together the main results relevant to biology, and
\cref{sec:6} reflects on the main points. For example, we see biological evolution as mirroring the human scientific approach: processes involving computation and testing, punctuated by something analogous to guesses.  This discussion also draws together some points that appeared as fragments in previous sections.

\section{Physics of computational networks}\label{sec:2}
We want a concept of computing that is applicable to a person with paper and
pencil, to a digital computer, or to a biochemical process in a living organism.
For this purpose, we extract some features from computer hardware---which,
unlike software, has a simple form---that we call a {\it computational
  network}.  In terms of hardware, a computational network consists of
interconnected logic gates, linked to an unpredictable environment.\footnote{1)
We are interested in finite networks of gates rather than in Turing
computability.  2) Some restrictions on networks to ensure proper functioning,
including repeatability, will be stated below.}  Each gate repeatedly performs
one or another logical operation. Partitioning logical operations into sequences
executed by logic gates imparts stability needed for memory. For example, memory
is often composed of cross-connected NAND gates that comprise flip-flops. As
mentioned in \cref{sec:one}, we allow for a wide range of gate mechanisms, including those that presumably occur in living organisms.

One can cartoon a computational network as a model railroad, with engines
carrying messages over tracks from station to station.  First, picture a
circular track with several stations located around it and a single engine.  When the engine arrives, a station modifies the message carried by the engine, puts the
modified message back on the engine, and sends the engine to the next station.  Now
picture a second circuit traversed by a second engine, with the two circuits sharing a
station.  The first engine that arrives at that station waits for the other engine
to arrive.  The messages from both engines are then processed together, and the
station generates a message for each engine to be conveyed to subsequent
stations on the two circuits.

As a cartoon for a computational network, the stations are logic gates, the
tracks connect the gates, and the engines carrying messages are tokens carrying
labels, with gates performing logical operations on the token labels.  At the finest level
of description each token label is 0 or 1.  Gates performing logical operations
on 0s and 1s can be connected in a network that adds and multiplies
numbers. Indeed, connected logical operations do all the computing that digital
computers and computer networks do.  We claim that this logical structure is
found not only in computer hardware, but also in all living organisms.  Below, we
discuss the ``ball-tossing'' aspect of logical synchronization necessary for the
logical operations of computational networks.

The physical mechanisms of the logic gates that perform logical operations  have a peculiar property.   Whether electronic or biochemical,   they depend on physically distinct conditions and on transitions among these distinct conditions.   What is peculiar is that {\it distinct} is undefinable in any way that is objective in the sense of being observer-independent. ``Distinct'' can mean distinct to one organism but not to others.  In the transitions among distinct conditions, logical synchronization is necessary to ensure that gates operate on distinct conditions rather than on an indistinct condition that must occur during  transitions. Crucial to the capacity of gates for distinguishing between conditions are  {\it tolerances of deviations}, so that tolerated irregularities do not cause errors.
For example,  in many digital computers, 0 is conveyed by a lower voltage while 1 is conveyed by a higher voltage, and to preclude errors due to unavoidable
voltage deviations, the high and low voltages have tolerated ranges separated by a ``dead band'' \cite{level}.

\subsection{Computing by logical operations in contact with unpredictable environments}\label{sec:2.1}
 The gates of a computational network {\bf physically perform} {\it logical
   operations} associated with {\bf mathematical} {\it logical connectives}
 such as {\it and}, {\it not}, {\it nor}, and {\it nand}, where {\it nand} is the
 negative of the {\it and} connective.  {\it Nand} is defined by the table of
 arguments and values in \cref{fig:1a}.  {\it We emphasize the distinction
   between {\it the physical performance of logical operations}, which involves
   motion, and {\it mathematical logical connectives} which can be written in
   tables that sit still on a page.}  To distinguish the operations from the
 connectives, we capitalize the operations as AND, NOT, NOR, NAND, etc.

 Because all the logical operations can be defined as compounds of NAND,
 that specific operation plays a major role in what follows.\footnote{How the AND, NOT, and NOR gates of computational networks are made from NAND
gates (together with FORKs) is shown in \cref{fig:1to3} in Appendix
\ref{sec:progControl}.}
 Each input of any logical
 operation can be either of two distinct physical conditions, as can the
 output.  These conditions are represented by 0 and 1, as indicated in
\vspace*{-8pt}
\cref{fig:1a}.
\begin{figure}[H]
  \hspace*{.3 in}
 \includegraphics[height=2.2 in]{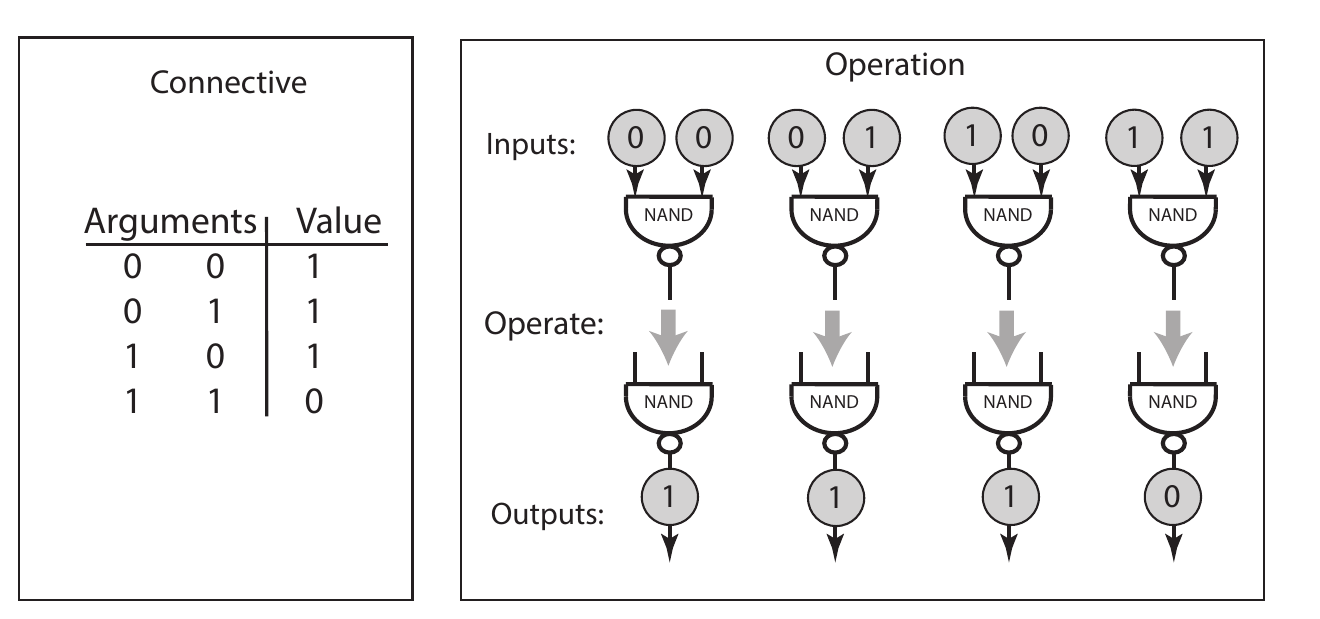}
 \caption{NAND connective and  NAND logical operation.\label{fig:1a} }
\end{figure}

\subsection{Gates and networks}
While the gates mentioned so far deal only with 0s and 1s, we also use the
term {\it gate} to indicate a mechanism that performs complex logical functions,
as occurs in coarsened descriptions of logical operations (see Appendix
\ref{sec:frag}).  We deliberately leave unspecified the various physical
mechanisms that can serve as gates, and we name gates after the operations that
they perform.

Anything a computer chip can do can be done by a network of NAND gates wired
together. To express the logic of wiring, we introduce a connective that we call {\it
  fork}, which as far as we can tell has been overlooked in formal
logic. \Cref{fig:1b} illustrates this connective along with the associated gate
operation FORK.
\captionsetup[figure]{font=small,skip=10pt}
\begin{figure}[H]\hspace*{.8 in}
 \includegraphics[height=2.5 in]{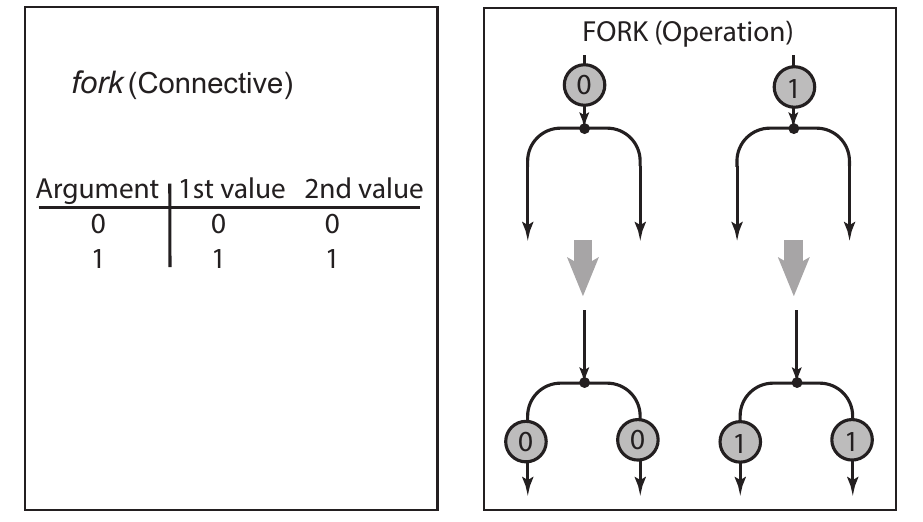}
 \caption{{\it fork} connective and FORK logical operation.\label{fig:1b} }
\end{figure}

 \subsection{Graphs and token games}\label{sec:tg} 
Graphs comprised of nodes connected by arrows are  used widely in science.  For example,
\cref{fig:citric}  illustrates the citric acid
cycle that powers most life \cite{cell}.
  \captionsetup[figure]{font=small,skip=10pt}
  \begin{figure}[H]\hspace*{1.2 in}
 \includegraphics[height=2.8 in]{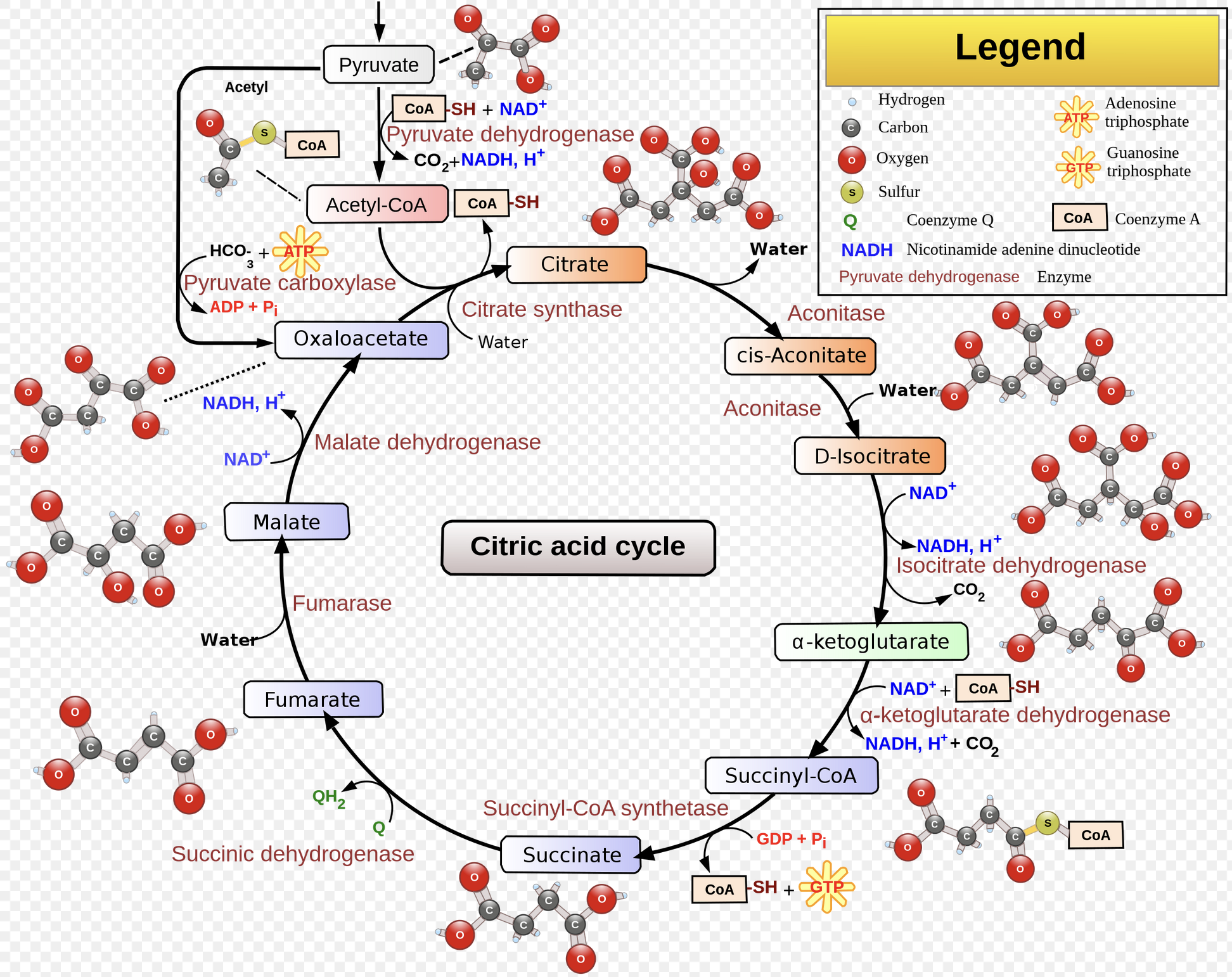}
 \caption{A graph of the citric acid cycle \cite{citric}.\label{fig:citric}
    }
  \end{figure}
  The concept of a computational network recognizes an unpredictable environment
 that generates inputs and receives outputs  (for process-control computations,
 this back-and-forth communication with an environment is evident).  We
 represent computational networks by graphs, in which each node 
 either represents a gate or is labeled ENV to represent an unpredictable environment that
 issues inputs and consumes outputs.  Each arrow points from one node to either the
 same or another node, representing a connection in which the output of one
 node serves as an input to that or another node.  This use of graphs leads to
 two developments, one physical and the other mathematical.  On the physical
 side, one can make a small computational network by using a drawing of the
 graph as a ``game board'' on which to either perform NAND and FORK operations by
 hand or make free choices to simulate the unpredictable ENV;
 tokens labeled by 0 or 1 are moved in what is called a {\it token game} \cite{96Petri}.
 On the mathematical side, representing computational networks by
 graphs leads to an elaboration of graph theory, i.e. {\it marked graphs}, which are discussed in \cref{sec:3}.

To play a token game  ``solitaire,''  a single person performs the logical operations for all the nodes. However, we mainly consider multi-player token games with  a player (human or otherwise) performing the role of each node. The rules of the game
are as follows.
\begin{enumerate}
\item A node is called {\it fireable} if there is a token on each arrow
  pointing into it.  When a node is fireable, a {\it firing rule} calls for the
  player of the node to remove the tokens on all the arrows pointing into the
  node, and {\it after} that to place a token on each of the arrows pointing out of
  the node.
\item {\it Labeling rules} (sometimes called coloring rules) specify how labels on output tokens depend on those on input tokens.  The labeling for NAND and FORK is specified in  figs.\ \ref{fig:1a} and \ref{fig:1b}, respectively.  An ENV node has no labeling rule; the player has a free choice to simulate the unpredictable  labeling by an environment.
\end{enumerate}
 
  To demonstrate  the need for the {\it after} in the firing rule, \cref{fig:vw}
shows snapshots of a token game, before, after, and during a firing.  \Cref{fig:vw}(a) shows a fireable node $v$, while
\cref{fig:vw}(c) shows the graph after a firing of $v$.
\Cref{fig:vw}(d) shows what the snapshot in the midst of the move could show, without the ``after'' in the firing rule: $v$ is no longer fireable but has
not been completely fired, and $y$ has been made fireable. Node $y$ can fire before the completion of the firing of $v$,
resulting in a return to the situation of \cref{fig:vw} (a), without node $w$
ever firing.  Preventing such behavior necessitates the
``after'' condition in the firing rule.
\begin{figure}[H]\hspace*{.5 in}
 \includegraphics[height=2.1 in]{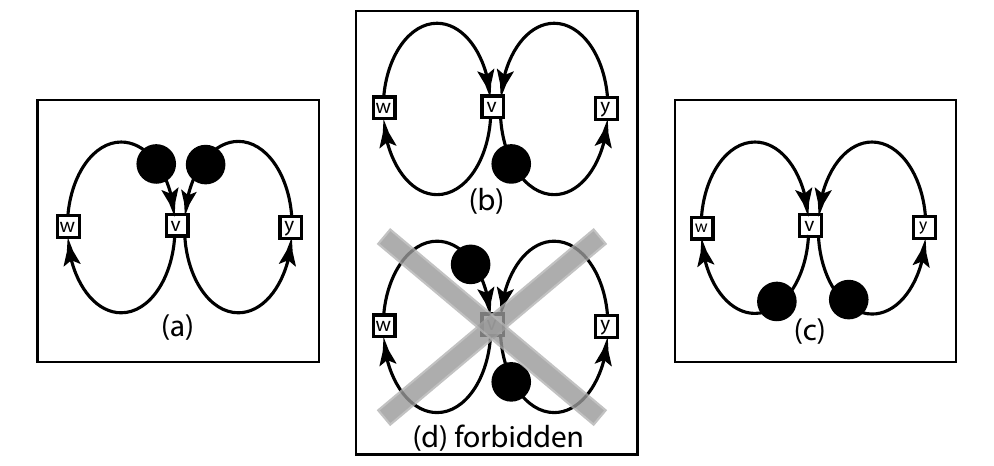}
 \caption{{\footnotesize (a) Node $v$ fireable; (b) node $v$ has been partially but not completely fired, making $y$ fireable but not $w$; (c) $v$ has been fired and $w$ and $y$ are concurrently fireable . Labels not shown.}
   \label{fig:vw}
 }
\end{figure}
  
\Cref{fig:first} shows a simple example of a token game:\footnote{The drawing in
\cref{fig:first} resembles a digital circuit diagram, from which we borrow the
symbol for NAND.}
  \begin{figure}[H]\hspace*{1 in}\vspace*{-.1 in}
 \includegraphics[height=2.2 in]{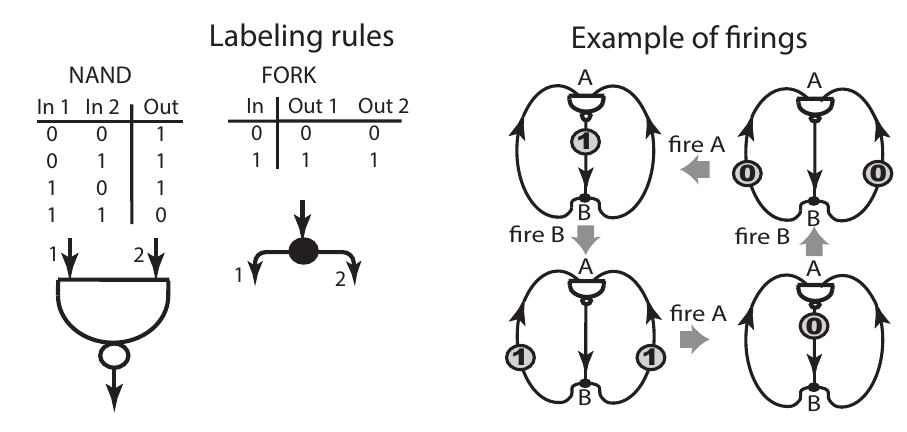}
 \caption{Example of token game.
   \label{fig:first}
 }
\end{figure}
 In a second example, a (simplified) thermostat performs a repetitive computation
to control the temperature of a room.  As illustrated in \cref{fig:therm1}, the
``Compute'' node compares a desired setting written on a token supplied by the
environment (ENV) with the temperature on a token from a digital thermometer in
the environment. From the result of this comparison, the Compute node sends a
command on a token to ENV to turn the furnace on or off.  The comparison makes use of an
arithmetic adder, illustrated in \cref{fig:full} in \Cref{sec:frag}.
  \begin{figure}[H]\hspace*{1.3 in}
 \includegraphics[height=1.8 in]{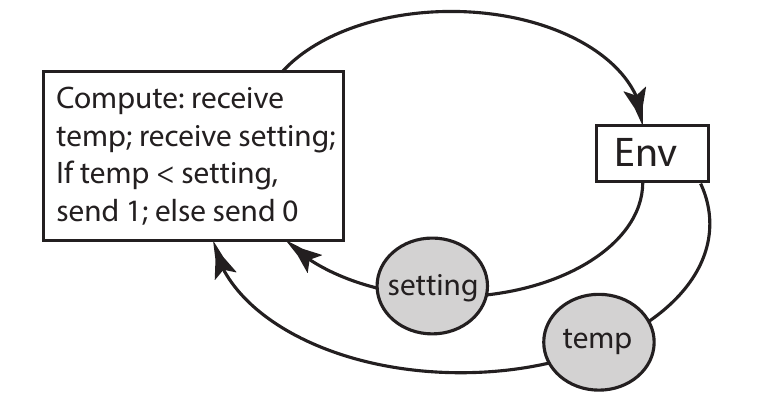}
 \caption{A token game for simple thermostat.\label{fig:therm1}
    }
\end{figure}

\subsection{The coordination of tokens depends on logical synchronization}\label{sec:logSync1}
Clocks step the gates of electronic computers through phases, and a gate
can receive a token in one phase but not in others \cite{14aop,19MST}.
Maintaining adequate phasing of token arrivals at logic gates constitutes
logical synchronization, illustrated in \cref{fig:2clock}.  Analogous phase
management is to be looked for in biochemical cycles  and the spiking of neurons.
\begin{figure}[H]
  \hspace*{.8 in}
 \includegraphics[height=1.5 in]{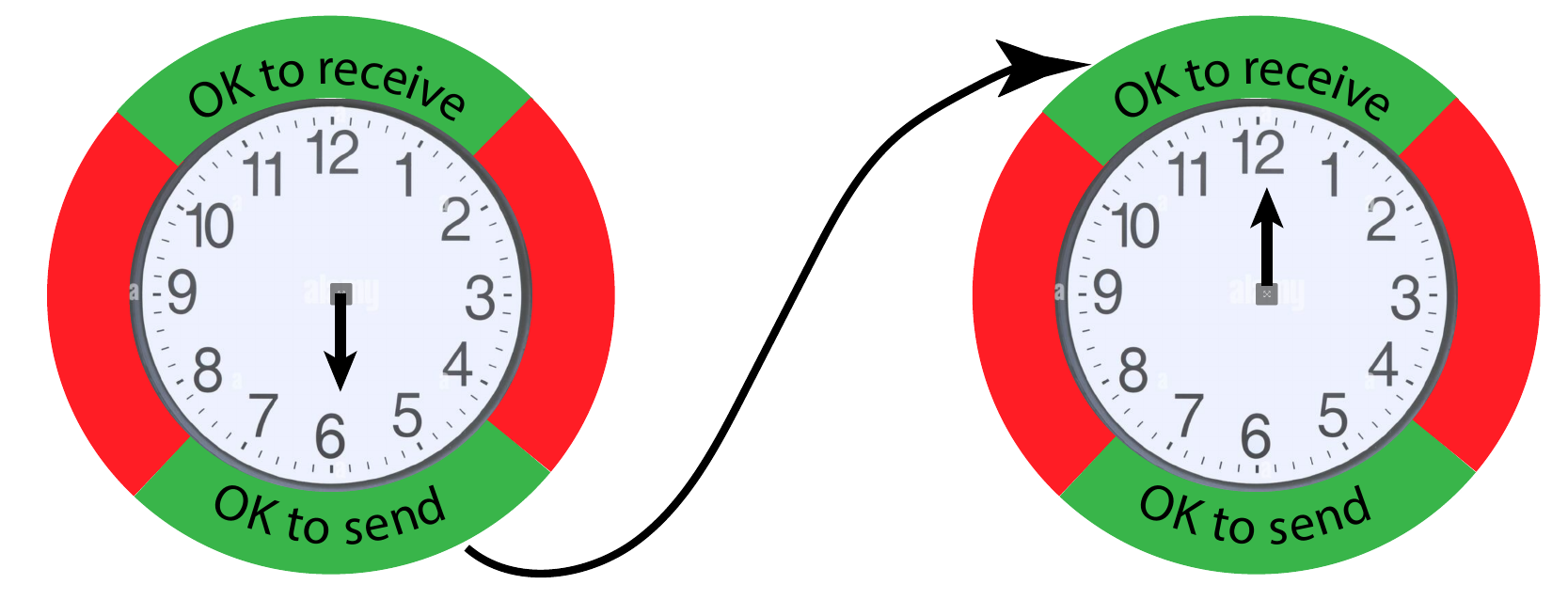}
 \caption{Phases for sending and receiving a token.\label{fig:2clock} }
\end{figure}
  The difference in hand positions of the two clocks indicates a propagation
 delay; e.g. the receiving clock at the arrow head shows a later time than the
 sending clock at the arrow tail.  As discussed in \cite{14aop,19MST},
 logical synchronization can be maintained between computers without them being
 set to standard time broadcasts, so long as their phases are aligned with
 propagation delays.
 
\subsection{Unpredictable steering to maintain logical synchronization}\label{sec:steer} 
Maintaining logical synchronization requires (1) detecting unpredictable deviations in the phasing of token arrivals, and (2) responding to those deviations by adjusting the rhythms of the computational networks to keep the deviations within allowed tolerances.
\begin{prop}
  The deviations of the phasing of tokens in computational networks are undetectable by those networks and must be detected by something else {\normalfont\cite{14aop,19MST}}.
\end{prop}
\begin{proof}
The logical functioning of computational networks must be indifferent to deviations in the phasing of token arrivals within their allowed tolerances.
 \end{proof}
\noindent Before it can be maintained, logical synchronization must be  acquired so that communication can begin.  It is often acquired quickly, but  there is no deterministic limit on the time required \cite[Chaps.\ 4 and 5]{meyr}.

\subsection{Disorder of computing from failures of logical synchronization}\label{sec:ff}
If logical synchronization fails in the transmission of a token from one part of
a network to another, the response of a receiving gate to an incoming
token label can be ambiguous, leading to a computer crash.  To reduce this
risk, designers interpose a ``decision element,'' based on the flip-flop shown
in \cref{fig:ffmm}.  However, this decision element suffers from a half-life of
indecision as discussed in \cref{sec:4.2}.  Thus, even when filtered by the
decision element, the ambiguity can propagate through a FORK to be seen as 0 by
one branch and as 1 by the other branch, thereby producing a logical conflict
known as {\it glitch}, which sometimes leads to a crash \cite{glitch}.

\subsection{A preview of cognition in living organisms} \label{sec:2.7}
In living organisms DNA codes are reproduced with astoundingly low error rates
available only in digital computing.\footnote{For example, ``the human mutation
rate \ldots is approximately 1 mutation/10$^{10}$ nucleotides/cell division''
\cite{cell}.}  While living organisms also make graded movements that correspond
to analog computation, without digital computation, their DNA replication would
drift enough to make life impossible \cite{life}.  Digital computation in living
organisms confronts a physical problem---the need for logical
synchronization---that human hardware engineers have solved well enough
that software engineers do not need to know about it \cite{ken}.  Many models of
living organisms employ logic gates, but to the best of our knowledge, we are
the first to recognize the need for logical synchronization in biological
computing.

Without discussing consciousness, we show here that computation has a place in human
control of physiological processes and human thinking. In \cref{sec:5}, we
model biological computation by computational networks linked to unpredictable
environments, understanding that the networks themselves undergo unpredictable
changes, and that maintaining logical synchronization requires something
beyond digital computing.  As represented by the marked graphs of
\cref{sec:revise}, changes in computational networks are pictured by
corresponding unpredictable changes in the marked graphs.  
We suggest that {\it logical distance} as defined in \cref{sec:3.3} serves as a
physiological underpinning of navigation.
 
\section{Representing computational networks by marked graphs}\label{sec:3}
Marked graphs have earned a special place in physics.  After five decades of
searching, they are the only mathematical objects we have found that deal with
motion without the extra-mathematical assumption of available physical clock
readings.  A peculiar and interesting advantage of marked graphs is this: while
motion is fundamentally beyond the expressive capacity of mathematics, a special
kind of motion essential to calculation and, indeed, to all living organisms,
exhibits logical relations expressed by marked graphs.

The mathematics of marked graphs discussed in this section enables the definition the important concept of logical distance.  Logical distance provides a bridge between the evolution of nervous activity in organisms and their capacity for spatial and temporal distinctions.  As described in  \cref{sec:4}, marked graphs also sharpen our understanding of computational networks.
See   \Cref{sec:G} for a tutorial on the relevant\vspace*{10pt}
mathematics.

\setlength{\fboxrule}{1.2pt}
\fbox{\begin{minipage}{32em}The marked graphs in this report represent computational networks designed for repeated runs of a single computation.  For this reason, our marked graphs differ from those that portray a general-purpose computer.
\end{minipage}}\vspace*{8pt}

Marked graphs mathematically represent the logical aspect of computational
networks.  They do not express continuously variable quantities such as voltages
and durations, but instead express {\it distinctions between physical conditions},
such as those that convey 0 or 1.  To arrive at the marked graphs used here, we
specialize Petri nets \cite{96Petri} to what in their jargon are called T-nets,
and we generalize them by allowing finite, directed multigraphs with loops (as
defined in \Cref{sec:G}).

Representing a computational network mathematically involves negotiating a gap
between physical motion and mathematical formulas.  As described below in
\cref{sec:4}, there need be no moment during which all nodes of a computational
network are between firings.  However, if the token game is played
``solitaire,'' then each possible sequence of moves has still moments during
which no nodes are firing, making unambiguous snapshots of token positions and
labels possible.  The locations of tokens at a moment between firings are
expressed mathematically by a {\it marking}, and a {\it labeled marking}
augments a marking by specifying the label carried by each token.  While labeled
markings are essential for representing computation, the {\it rhythms} of
logical operations can be expressed by unlabeled markings.

We restrict our marked graphs to markings that always make some node fireable and that never put more than one token on an arrow;  these are called {\it live and safe markings}.  More precisely, a marking is {\it live} if every node is fireable or can be made fireable through some sequence of firings.  A live marking is {\it safe} if it puts no more than one token on any arrow and if no sequence of firings starting from that marking can put more than one token on any arrow \cite{71MarkedGraph}.
Proofs of propositions in this section draw on 
\Cref{sec:G}, which is based on \vspace*{6pt} \cite{71MarkedGraph}.

The mathematical expression of the {\it firing} of a node is a relation between
two live and safe labeled markings: $M$ and $M'$ are related by the firing of
node $v$: (1) if $M$ specifies a token for each in-arrow of node $v$ and no
token for each out-arrow, if $M'$ specifies no token for each in-arrow of
$v$ and a token for each out-arrow, if $M$ and $M'$ agree about the presence of tokens on arrows not in contact with $v$, and (2) if the labels on the tokens on
out-arrows of $v$ follow a labeling rule for node $v$.
Mathematically, {\bf a firing is a relation};
nothing moves.  The ``snapshots'' $M$ and $M'$ exemplify
the mathematical representation of the logical relations inherent in logically synchronized motion.

Live and safe markings are possible only for a graph that is covered by
circuits.  \begin{prop}\label{prop:circ} The number of tokens on any circuit is
  the same for any pair of markings related by a sequence
  of firings.
\end{prop}
\noindent \Cref{prop:circ} implies that the possible live 
markings of a graph partition into equivalence classes, and following \cite{71MarkedGraph} we call these classes
of markings {\it families}.  In discussing families, we pay no attention to labels on tokens.
We then \vspace*{5pt}have the following proposition.
\begin{prop}\label{prop:fams}
  The number of tokens on any circuit 
  is the same for all markings in any family \textup{\cite{71MarkedGraph}}. 
\end{prop}

\subsection{Multiple families of markings}\label{sec:3.1}\vspace*{-10pt}
\captionsetup[figure]{font=small,skip=0pt}
\begin{figure}[H]
 \includegraphics[height=2.3 in]{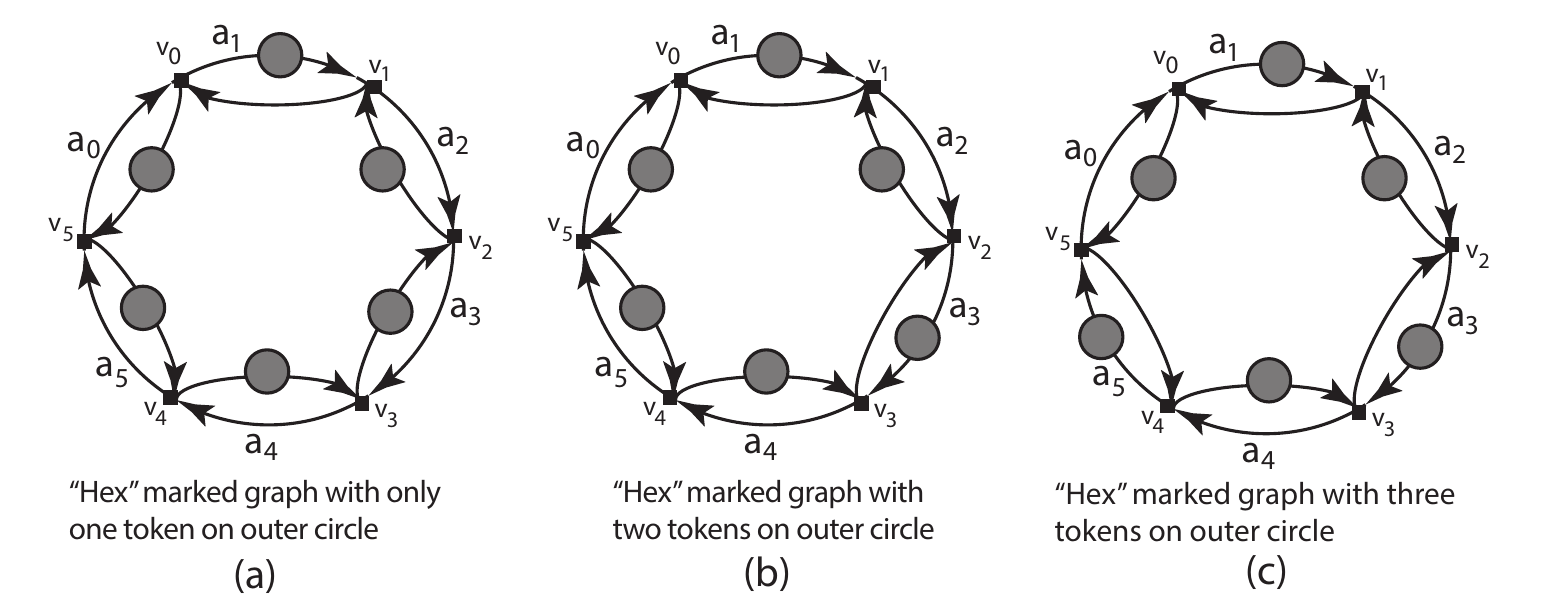}
 \caption{``Hex'' graph showing 3 markings in distinct families. \label{fig:Hex3} }
\end{figure}
\noindent As illustrated in \cref{fig:Hex3}, some graphs admit more than one
family.  The three markings belong to three distinct families, as follows from \Cref{prop:fams}, because the token count (i.e., the count of tokens) of the outer
(and also the inner) circuit differs in the three cases. The simplest such marked
graph is shown in \cref{fig:2triangle}.
  \begin{figure}[H]\hspace*{1.4 in} 
 \includegraphics[height=1.5 in]{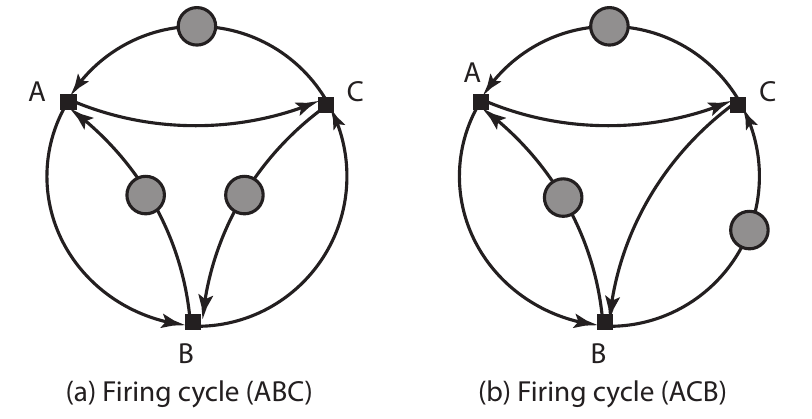}
 \caption{\label{fig:2triangle}Families with opposing cyclic orders of firing.}
\end{figure}
\noindent{\bf Remark}:  The same topology as that of \cref{fig:2triangle}
appears in the flip-flop
shown in \cref{fig:ffmm}, with its two live and safe families and
its two simple circuits that put nodes in opposite cyclic orders, i.e., (A, ENV, B) and (B, ENV, A).  An
interesting question is whether digital designs can avoid multiple marking
families, and if so, at what cost in terms of the number of logical operations required.
\captionsetup[figure]{font=small,skip=10pt}
\begin{figure}[H]\hspace*{1.2 in}
 \includegraphics[height=1.3 in]{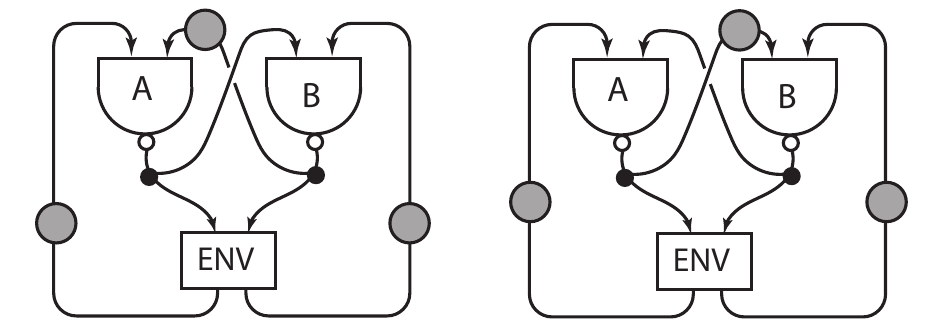}
 \caption{ Markings belonging  to different families. \label{fig:ffmm} }
\end{figure}

The markings in a family can be viewed as vertices of what is called a state-transition graph (we say vertices rather than nodes to distinguish state-transition graphs from graphs for computational networks). Exactly  when a marking $M_2$ can be reached from a marking $M_1$ by firing a node of the  marked graph, the state-transition graph has an arrow from $M_1$ to $M_2$.  Examples for two families are shown in \cref{fig:HexMark4,fig:Hex2B}.
\captionsetup[figure]{font=small,skip=0pt}
\begin{figure}[H]\hspace*{.2 in}
 \includegraphics[height=3.1 in]{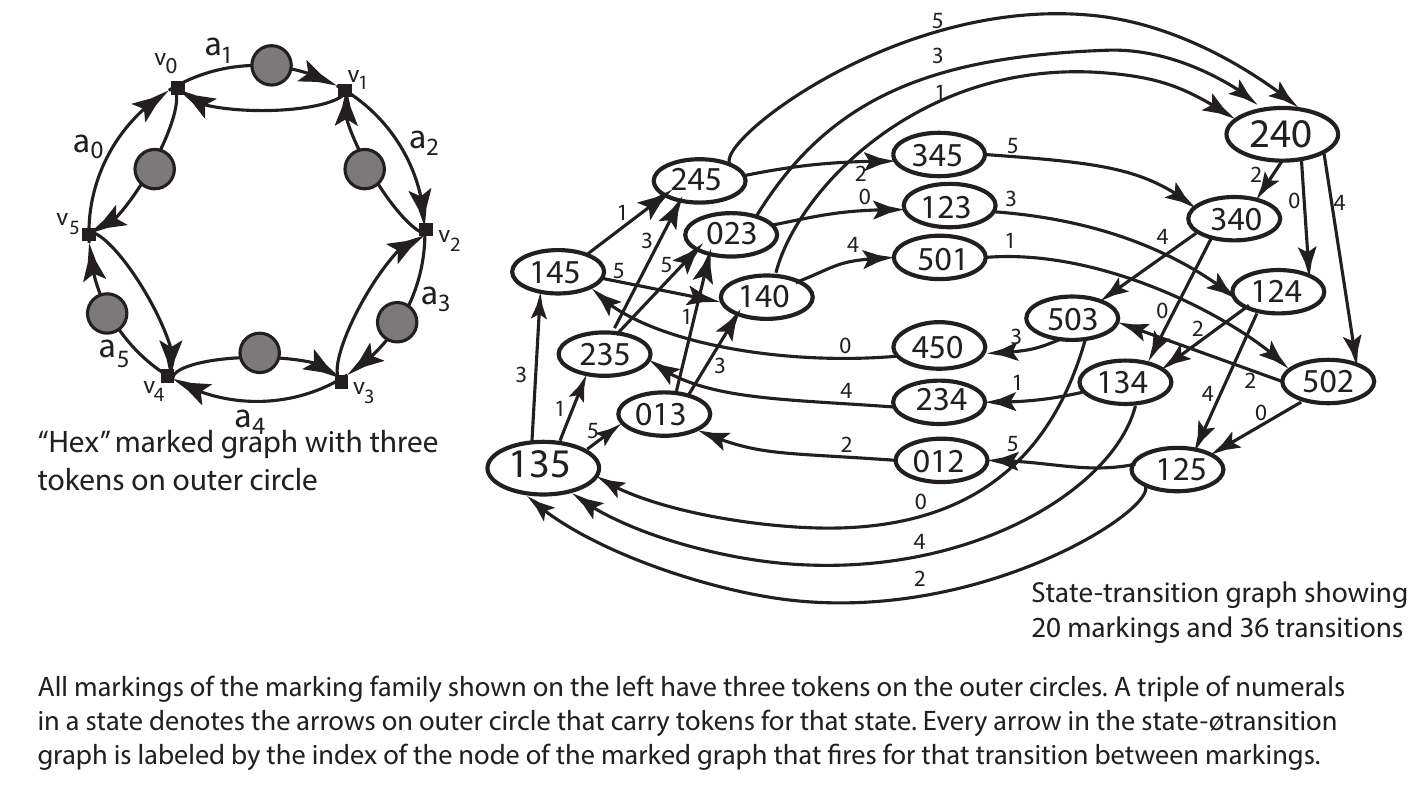}
 \caption{\label{fig:HexMark4}State-transition graph for ``Hex'' with {\bf 3} tokens on its outer circuit.}
\end{figure}
\begin{figure}[H]\hspace*{.81 in}
 \includegraphics[height=2.9 in]{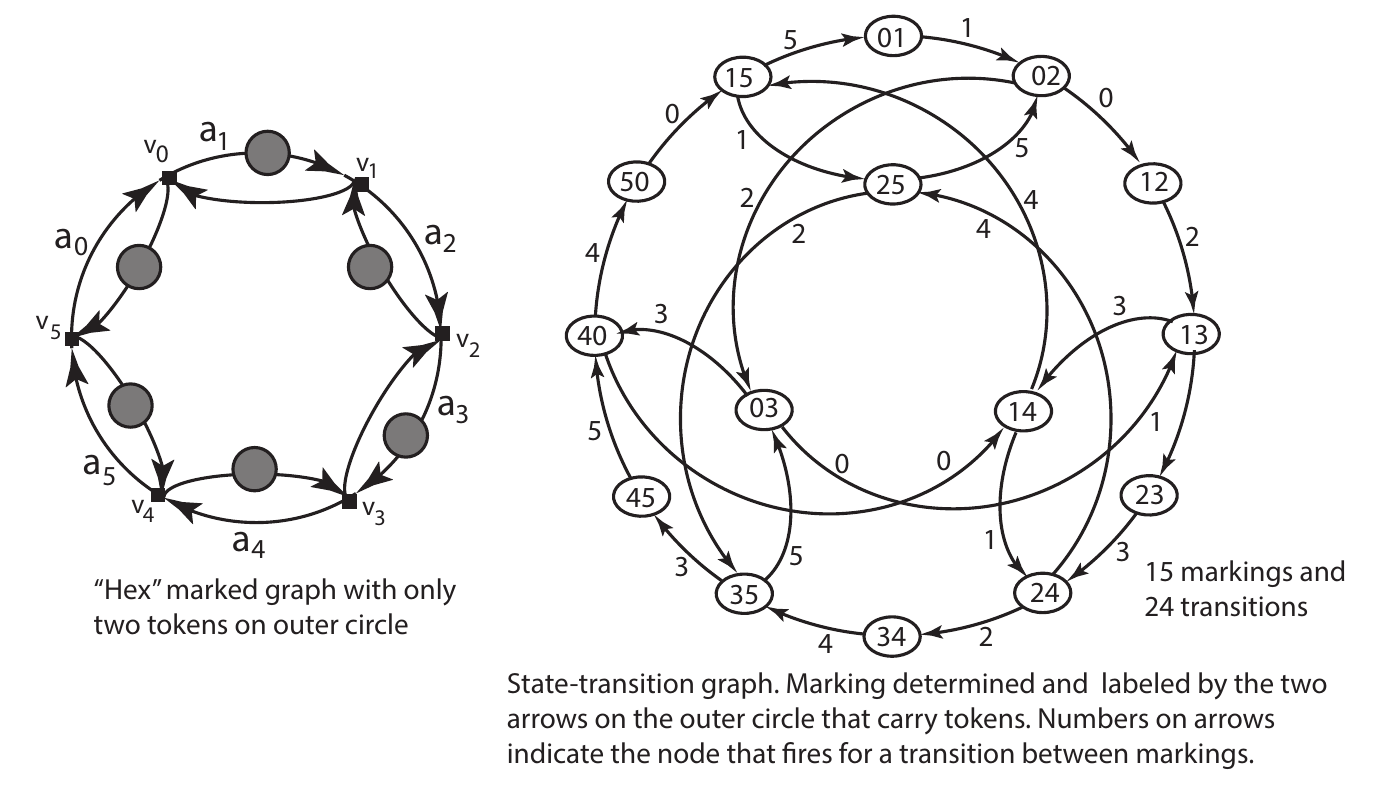}
 \caption{\label{fig:Hex2B} State-transition graph for ``Hex'' with {\bf 2} tokens on its  outer circuit.
 }
\end{figure}
\noindent For both families,  \cite[Theorem 7]{71MarkedGraph}  implies that the state-transition graphs have simple directed circuits through firings of all the nodes of the corresponding marked graph.
\Cref{fig:HexMark4,fig:Hex2B} each show several such circuits.

In many cases, different families for the same graph represent
computations of different functions.  The existence of graphs that support more
than one family shows that a family conveys information not
conveyed by the underlying graph.  One can ask: what
conditions on a strong graph are necessary and sufficient for the existence of
more than one family for that graph?  See the next subsection for
part of the answer, which we have not seen elsewhere.

\subsection{Conditions for the existence of multiple families}\label{sec:3.5}
The issue now arises of characterizing graphs that admit multiple families vs. those
that admit only one family.  To demonstrate a necessary condition for multiple
families, we must first establish a necessary condition for a simple directed
circuit $C$ in a graph $G$ with a live and safe family to have a token count of
$t>1$.  This condition involves basic circuits, which are simple directed circuits
with token count of 1 for the family considered. For safety, each arrow in $C$ must
also belong to a basic circuit in $G$.  Hence $C$ must be covered by a set
of at least $t$ basic circuits, and each basic circuit must contain an
arrow of $C$ not contained by any other basic circuit.  We call a covering with
this property {\it proper}.  For the $G$-marking to be live as well as safe, the
proper covering of $C$ must include  $m\ge t+1$ basic circuits, so that
at least one of them can be marked on an arrow not in $C$, making at least one node of $C$ fireable.
For any proper covering by $m$ basic circuits there is at least one set of  $m$
nodes of $C$ such that each node belongs to a distinct pair of basic circuits.  Let $S$ be such a set. (In some cases $C$ contains nodes not in $S$, but this does not matter.)
  The $m$ nodes of $S$ are cyclically ordered by the arrows of $C$.
Running counter to $C$, there is a simple directed circuit contained in the
proper covering, the arrows of which go through the nodes of $S$ in the cyclic
order opposite to that for $C$.  We call this opposing circuit $C_{\text{opp}}$.
To summarize, we have
\begin{prop}\label{prop:prep}
If $C$ is a simple directed circuit with token count $t>1$ in a graph $G$ with a live and safe family, then:
\begin{enumerate}
\item $G$ must contain a proper covering of $C$ by  $m \ge t+1$ basic circuits,
\item $C$ contains a set $S$ of $m$ nodes of $C$ that are cyclically ordered by arrows of $C$. 
\item The graph $G$  contains a second simple directed circuit $C_{\text{opp}}$ containing the $m$ nodes of $S$ cyclically ordered in the direction opposite to that of $C$.
\end{enumerate}
\end{prop}
If a pair of simple 
directed circuits cyclically order 3 or more nodes in opposite directions, we call them {\it opposing} simple directed circuits.
We can now state a condition necessary for a graph $G$ to contain more than one live and safe family.

\begin{prop}\label{prop:yes}
If a graph admits more than one live and safe family, then the graph contains opposing simple directed circuits. 
\end{prop}

\begin{proof}
Suppose that a graph admits two distinct live and safe families, say
$\mathcal{M}$ and $\mathcal{M}'$.  By Theorem 12 of \cite{71MarkedGraph}, one
family, say $\mathcal{M}$ must place more tokens on some simple directed circuit
$C$ than does the other family $\mathcal{M}'$.  By Theorem 1 of
\cite{71MarkedGraph}, $\mathcal{M}'$ must place at least one token on $C$, and
so $\mathcal{M}$ must place at least two tokens on $C$.  By \Cref{prop:prep},
this implies that the graph contains opposing simple directed circuits. (An alternative proof can be found in \cite{80Murata}.)
\end{proof}

 \noindent {\bf Question}:
Does the converse of \Cref{prop:yes}
hold; i.e., does every strong graph with 
opposing simple directed circuits admit more than one live and safe family?  By filling in an omission in a proof  in \cite{80Murata}, we find the following.
\begin{prop}\label{prop:koh}
  If a strong directed graph $G$ contains opposing directed circuits it admits more than one live and safe family
\end{prop}
\begin{proof}
  Let $M(C)$ denote the token count of a directed circuit $C$ under a marking
  $M$, and recall that we speak of a directed circuit $C'$ in $G$ as basic in
  $M$ if and only if $M(C')=1$. If a strong graph $G$ contains opposing simple
  directed circuits, their covering by basic circuits implies that for any live
  and safe marking $M$ at least one of the opposing simple directed circuits,
  say $C$, has token count $M(C) \ge 2$. Choose an arrow $e$ in $C$ such that
  $M(e)=1$. Since $M$ is safe, $e$ belongs to one or more directed circuits
  basic in $M$. Construct a marking $N$ without any token on $e$ and with 1
  token elsewhere on each directed circuit that is basic in $M$ and contains
  $e$. Clearly $N$ is live. If $N$ is safe, it follows that $G$ admits more than
  one family since $M(C) >N(C)$.  If $N$ is not safe, then by the argument in
  the proof of Theorem 2 of \cite{71MarkedGraph}, there is a sequence of firings
  and token removals that make a live and safe marking $N'$ with only one token
  on $C$.  Then $M(C) > N’(C)$ and $M$ and $N’$ are live and safe markings in
  different families.
  \end{proof}

 \subsection{Logical distance}\label{sec:3.3}
Marked graphs with live and safe families permit the definition of a logical analog of the concept of distance.  The notion of distance used in physics has evolved in the quest for increasing precision.  Rather than taking a meter inscribed in platinum as the defining standard, the meter is now defined by invoking the theory of relativity in which the speed of light allows a new definition: the meter is how far light travels in a specified fraction of a second \cite{SI}.  This definition is inspired by the theory of special relativity, which relates distance to a count of clock ticks between the sending a signal and the receipt of an echo from a distant point \cite{05einstein} (with a practical  implementation in radar \cite{radar}). We define an analogous, purely mathematical {\it logical distance} between nodes of a marked graph with a live and safe family.

\begin{defn}\label{defn:onex}
  Assume that the firing of a node can copy a label on an in-arrow of the node to a label on an out-arrow of the node. Then for any two nodes $v$ and $w$ of a graph with a live and safe family, the {\it logical distance} $D(v,w)$ is the number of times $v$ must fire in order for a label carried by a token on an input arrow of $v$ to reach $w$ and be returned as a label on a token on an input arrow to $v$.  By convention, if $v$ and $w$ are identical, then $D(v,v) = 0$.\footnote{{\it The logical distance} is distinct from {\it path length} in graphs.}
\end{defn}

\begin{prop}\label{prop:ec1}
  Suppose that $v$ and $w$ are any two nodes of a connected graph with a live and safe family.  The graph contains at least one simple directed path from $v$ to $w$ and at least one return simple directed path from $w$ to $v$.  For a marking in the family, let $n_{vw}$ be the least token count of any path from $v$ to $w$, and let $n_{wv}$ be the least token count of any return path from $w$ to $v$.  The logical distance $D(v,w)$ is the sum $n_{vw}+n_{wv}$ of the two  token counts. {\normalfont(See example of nodes $v_0$ and $v_5$ in    \cref{fig:nCy} below.)}
\end{prop}

\begin{proof}
  By \Cref{prop:handy} in the Appendix, the sum $n_{vw}+n_{wv}$ depends only on the family and not on the choice of marking within a family. For a label to progress across a node, the node must fire and copy the label from a token on an in-arrow to a token on an out-arrow.  Because the marking is safe, no token can catch up with any other token, so the label cannot be returned to $v$ until $v$ has fired $n_{vw}+n_{wv}$ times.   Furthermore, after this number of firings of $v$, the label can return to $v$, perhaps via firings of nodes other than $v$, but without any more firings of $v$.
\end{proof}

\begin{figure}[H]\hspace*{.8 in}
  \includegraphics[height=1.9 in]{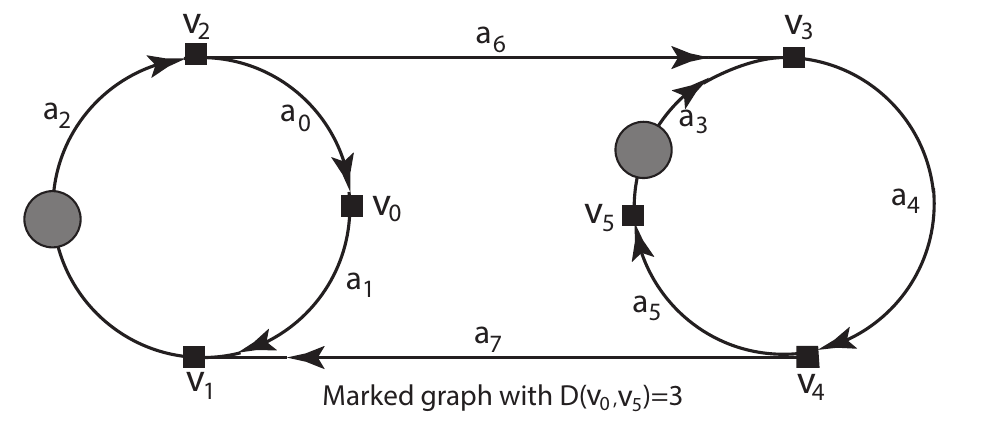}
 \caption{Directed path from $v_0$ to $v_5$ overlaps the directed path from $v_5$ to $v_0$ in arrows $a_2$ and $a_4$.  \label{fig:nCy} }
\end{figure}
In \cref{fig:nCy}, the token count of the path from $v_0$ to $v_5$ is 1, while the token count of the path from $v_5$ to $v_0$ is 2, so the logical distance is 3; the token on arrow $a_2$ is on both paths and is thus counted twice.
\Cref{fig:nCy2} shows an added path that eliminates the need for the overlapping paths  of  \cref{fig:nCy} and  reduces the  logical distance to 1.
\begin{figure}[H]\hspace*{.8 in}
 \includegraphics[height=1.6 in]{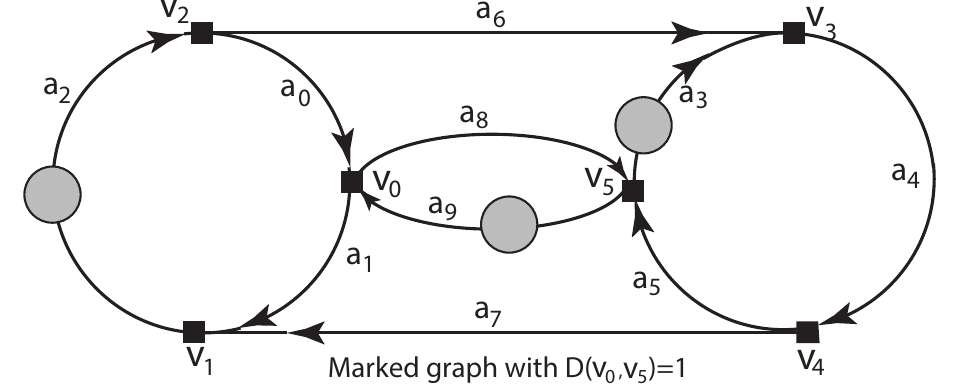}
 \caption{Graph with added path without overlap.  \label{fig:nCy2} }
\end{figure}

\begin{prop}\label{prop:ec2}
  The  logical distance $D$ is a {\it metric} on nodes of a  graph with a given live and safe family.
\end{prop}
\begin{proof} For any connected graph with a live and safe family,
\Cref{defn:onex} makes $D(v,v)=0$ and makes the logical distance between
two distinct nodes positive. By the argument in the proof of
\Cref{prop:ec1}, the logical distance is invariant under swapping the two nodes
and is thus symmetric.  The triangle inequality follows from concatenating bi-directional pairs of directed paths; thus the defining conditions of a metric are fulfilled.
\end{proof}

\subsubsection{The logical distance depends on the family}

\begin{prop}\label{prop:ec3}
  For a given graph, the logical distance between two nodes can depend on the family. 
\end{prop}
\begin{proof}[\nopunct]{\it Proof  by example}: By \Cref{prop:fams}, the markings shown
in \cref{fig:Hex3} belong to distinct families.  The family that puts two tokens
on the outer circuit leads to $D(v_0,v_3) = 2$, while the family that puts three
tokens on the outer circuit leads to $D(v_0,v_3)=3$.
\end{proof}

\subsection{Coarse representations with generalized nodes and labels}\label{sec:3.2}

One often wants to suppress detail.  At the finest level of representation, our marked graphs contain only the simple nodes NAND, FORK, and ENV,
with token labels 0 or 1, and such graphs can be huge and confusing.  To suppress detail, one ``modularizes'' a graph by contracting fragments of it into  single nodes, resulting in a coarser level of representation.\footnote{Contractions of graphs are a
special case of graph morphisms.  In contrast to an unmarked graph, contracting
a marked graph requires paying attention to the effect of the contraction on markings
and their labels.}  After contracting a fragment to a node of a coarser graph,
the labeling rule for the new node can define a complex logical operation on
strings of bits rather than just single bits.

When focusing on the motion of logic, we sometimes omit token labels;  \cref{fig:condense1} shows contractions
with token labels omitted.
\begin{figure}[H]\hspace*{.5 in}
 \includegraphics[height=1.7 in]{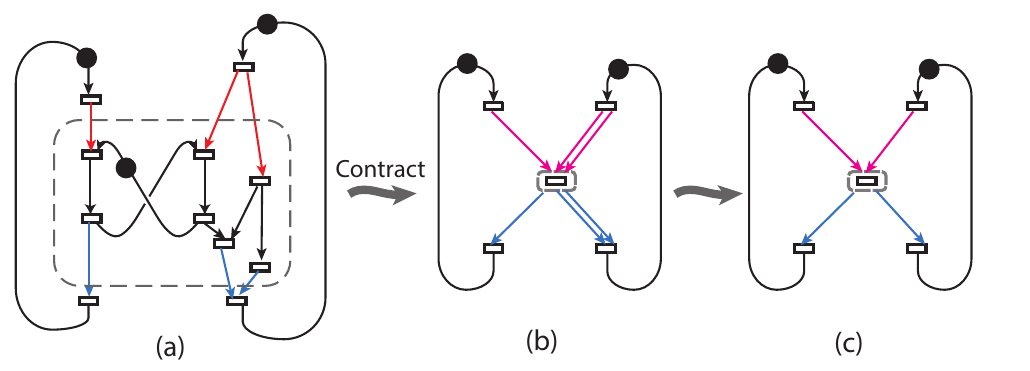}
 \caption{(a) A marked graph with a subgraph within the dashed boundary;  (b) the contraction of a subgraph to a node; (c) further contraction.\label{fig:condense1}}
\end{figure}
\noindent Labeling rules for contractions can involve memory in addition to
inputs and outputs.  Allowing multigraphs with loops, rather than the
simple graphs usual for Petri nets, simplifies contractions.  
\Cref{sec:progControl} uses contractions in showing how
marked graphs can express the ``if-then-else'' construct used in computer programs.

\subsection{Representing changes in computational networks}\label{sec:revise} 
 Per ``no final answers,'' surprises lead to revisions in human-made computational networks, and analogous needs for revision are endemic to computational networks found in all living organisms.  Although surprises defy mathematical prediction, a change in a computational network in response to a surprise can be expressed mathematically as a change in a live and safe marked graph.

 Changes in live and safe marked graphs can  express evolutions of the equations of physics, including drastic changes, as occurred with the introduction of quantum theory. Equations used in predictions and other computations are inscribed in computational networks, and a change in the equation used implies a change in the logic of the computational network and, hence, a change in the marked graph that expresses that  logic.  Therefore, changes in equations of physics, as well as changes in the logical  structure by which all living organisms function,  can be represented by changes in marked graphs with live and safe markings.  The suitability of marked graphs as mathematical tools for expressing such changes is a major  motivation for our interest in them.

 \Cref{fig:xmit} illustrates two types of change: (1) the fusion of nodes of different computational networks into a single node to create a single computational network and (2) the fission of a single node along with arrow additions to create a more complex computational network.
 \begin{figure}[H]
  \hspace*{1 in}
 \includegraphics[height=4.3 in]{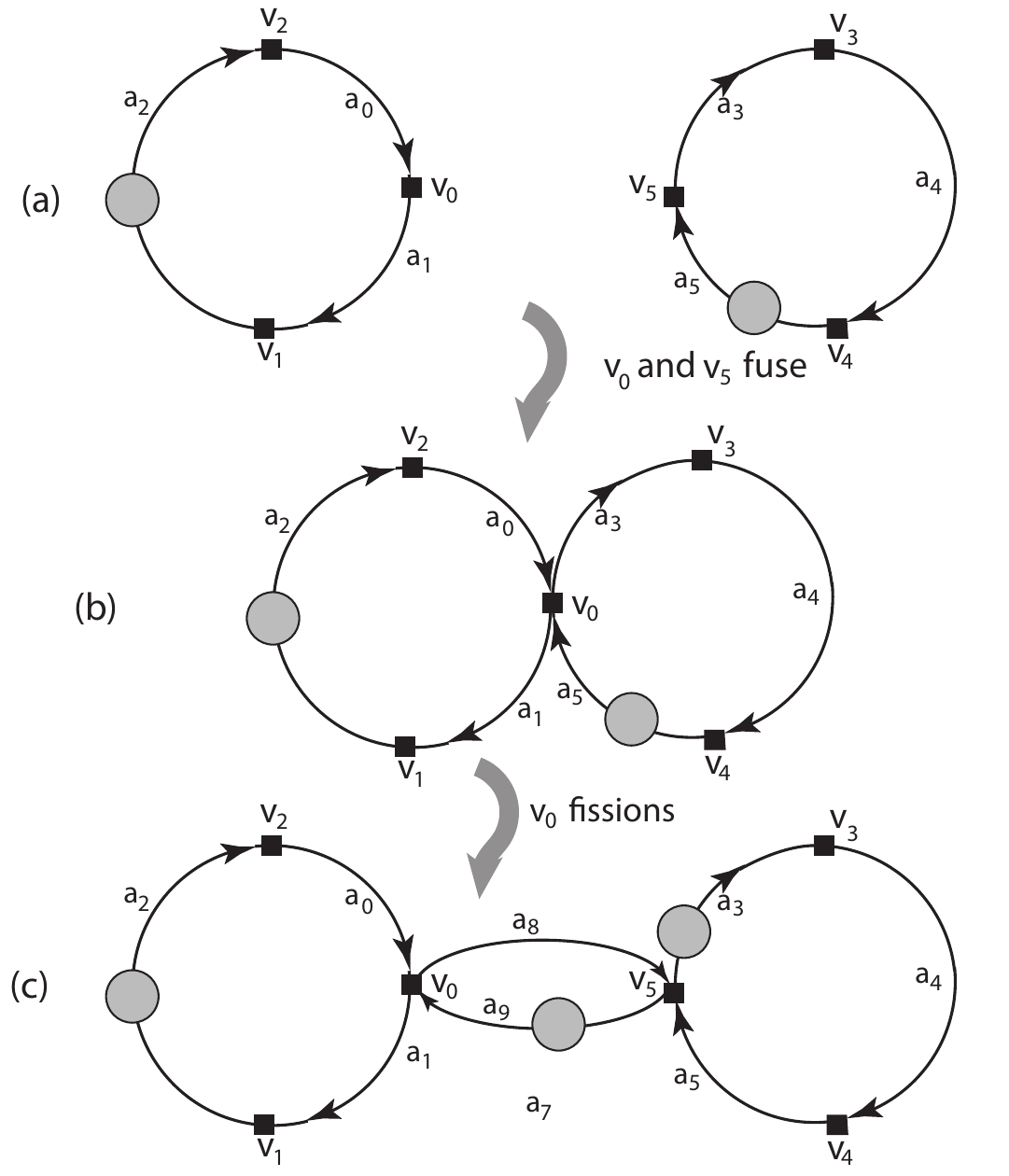}
 \caption{Fusion and fission of computational networks.\label{fig:xmit} }
\end{figure}
\noindent 

While changes in graphs without token games have been well studied \cite{edit}, the relevant changes correspond to operations that
\begin{enumerate}
\item enlarge a graph $G$ to make a new graph $G'$ {\it in such a way that the live
  and safe family for $G$ is a restriction of the live and safe family for
  $G'$}, or 
  \item eliminate nodes and/or arrows of a larger $G'$ to make a smaller $G$ with a live and safe marking derived from that of $G'$.
\end{enumerate}
Some such changes have been discussed \cite{80Murata,81Murata}.
We hope to report later on other types of changes that preserve liveness and safety. In \cref{sec:locAssemble} below we note a related question of the availability at a node of  information indicating a change  in the graph near that or another node.

\section{Marked graphs and the local physics of computation}\label{sec:4}
Considering {\it computation as measurable physical behavior}, we ask the following question. What evidence from measurements of a computational network can determine the graph and the family that represent it?  A first thought regarding obtaining evidence could be to take snapshots of the computational network and relate those snapshots to markings of the graph, as in \cref{fig:first}. However, there is an obstacle, which was alluded to in \cref{sec:3}. Although a sequence of unambiguous snapshots can record a chess game, the absence of ambiguity depends on moments between moves when all the pieces sit still on the game board.  In computational networks, while every individual gate performs sequentially, there need be no moment when all gates are between moves, as illustrated in \cref{fig:Hex}.
\begin{figure}[H]\hspace*{.8 in}
  \includegraphics[height=4.3 in]{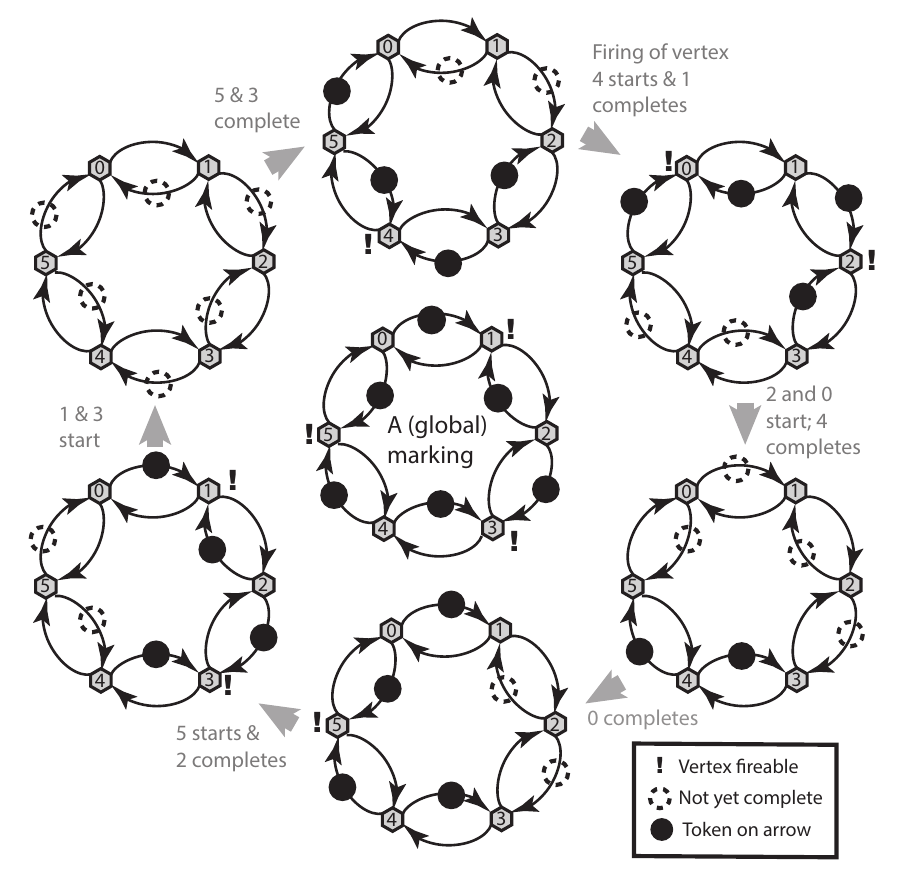}
 \caption{A token game that allows no global snapshots.\label{fig:Hex}
    }
\end{figure}

The center of \cref{fig:Hex} shows a possible marking of the ``Hex'' graph, but
the surrounding snapshots illustrate a play of a token game in which some node
is always in the midst of firing so that the token positions never correspond to
any marking, even though the node firings conform to the firing rule.  The
marking at the center of \cref{fig:Hex} makes three nodes {\it concurrently
  fireable}, obstructing any correspondence between markings and snapshots
of the physical network.  However, a correspondence between such a computational
network and its representation by a marked graph can be demonstrated through the use of a
suitably coordinated set of {\it local snapshots}, as we now discuss.

\subsection{When global snapshots fail, local records work}\label{sec:locAssemble}
When global snapshots fail, local records can be organized to determine the
graph and the family for a computational network, at the expense of endowing nodes with additional functions.  For example, let each node
make a record each time it fires.  The record includes a count of the node's own
firings and its own name.  The node writes these records on the labels of the
tokens that it produces so that its own record of a firing can include the
names and counts that come in the labels on its input tokens.  Without the
counts, local records (once collected) determine the graph but not the family.
With the counts, the marking class can be determined. For the example, local records suffice to distinguish among
the three families shown in \cref{fig:Hex3}.\footnote{When more than one arrow
connects two nodes, the arrow names must also be included in the local records.}
As an explicit example, we call the family for which the outer circuit has a token count of 2
and the inner circuit has a token count of 4 the {\it 2-4 family}, and similarly for
the 3-3 \vspace*{8pt}family.

{\bf \hspace*{.753 in}
  2-4 family \hspace*{1.07 in} $\Big|$ \hspace*{.75 in} 3-3 family}\\
\setlength{\textwidth}{2.75 in}
\noindent\fbox{\begin{minipage}{\textwidth}
    $V_1$\#1; source $V_2\#0, V_0\#0$; target $V_2,V_0$
\end{minipage}} $\Big|$ 
\fbox{\begin{minipage}{\textwidth}
    $V_1$\#1; source $V_2\#0, V_0\#0$; target $V_2,V_0$
\end{minipage}}\\
\setlength{\textwidth}{2.75 in}
\noindent\fbox{\begin{minipage}{\textwidth}
    $V_3$\#1; source $V_2\#0, V_4\#0$; target $V_2,V_4$ 
\end{minipage}} $\Big|$ 
\fbox{\begin{minipage}{\textwidth}
    $V_3$\#1; source $V_2\#0, V_4\#0$;  target $V_2,V_4$
    \end{minipage}}\\
\setlength{\textwidth}{2.75 in}
\noindent\fbox{\begin{minipage}{\textwidth}
    $V_2$\#1; source $V_1\#1, V_3\#1$; target $V_1,V_3$
\end{minipage}} $\Big| $
\fbox{\begin{minipage}{\textwidth}
 $V_2$\#1; source $V_1\#1, V_3\#1$; target $V_1,V_3$
\end{minipage}}\\
\setlength{\textwidth}{2.75 in}
\noindent\fbox{\begin{minipage}{\textwidth}
    $V_4$\#1; source $V_3\#1, V_5\#0$; target $V_3,V_5$ 
    \end{minipage}} $\Big|$ 
\fbox{\begin{minipage}{\textwidth}
$V_4$\#1; source $V_3\#1, V_5\#1$; target $V_3,V_5$ 
    \end{minipage}}\\
\setlength{\textwidth}{2.75 in}
\noindent\fbox{\begin{minipage}{\textwidth}
    $V_5\#1$; source $V_4\#1, V_0\#0$; target $V_2,V_0$ 
\end{minipage}} $\Big|$ 
\fbox{\begin{minipage}{\textwidth}
 $V_5\#1$; source $V_4\#0, V_0\#0$; target $V_2,V_0$
\end{minipage}}\\
\setlength{\textwidth}{2.75 in}
\noindent\fbox{\begin{minipage}{\textwidth}
    $V_0$\#1; source $V_5\#1, V_1\#1$; target $V_5,V_1$ 
    \end{minipage}} $\Big|$
\fbox{\begin{minipage}{\textwidth}
     $V_0$\#1; source $V_5\#1, V_1\#1$; 
    target $V_5,V_1$
\end{minipage}}\vspace*{8pt}

\noindent  The local records link markings logically, not temporally.  The logical links determine the families, shown in the state-transition graphs of figs.\  \ref{fig:HexMark4} and \ref{fig:Hex2B}.

Whether a computational network is human-made digital or an evolving biological network, a change that is significant locally may or may not be significant at locations some radar distance away.
As an avenue to explore, we are interested in the question: how can records available at one node reflect changes made at other nodes,whether  close by or more distant?

\subsection{The physical undecidability of the temporal order of concurrent events}\label{sec:4.2}

Deciding physically on the temporal order of concurrent events has a known
trouble that challenges the very concept of a time coordinate that can be valid
over a computational network.  The trouble is the occasional inconsistent
assignments of temporal order. For example, in the computational network
represented by \cref{fig:Hex}, there can be three nodes firing more or less at
once, corresponding to the three concurrently fireable nodes.  Determining
temporal order is then the judging of a three-way race.  As  shown in \cite{19MST}:
\begin{enumerate}
\item Physically measuring the order of occurrence of three events $A$, $B$, and $C$---such as three token firings---requires measuring pairwise orders, i.e., between $A$ and $B$, between $A$ and $C$, and between $B$ and $C$.
\item Because of uncertainty in pairwise decisions, attempts to assign temporal order to three concurrent firings can violate the transitivity of an order relation by outcomes of $A$ before $B$, $B$ before $C$, but $C$ before $A$ rather than $A$ before $C$.
\end{enumerate}

Another factor that complicates the physics of temporal ordering is  the ``time required to decide.''     The mechanical and electronic devices used to decide ``which came first''  depend on {\it balancing one effect against another}. These devices can be built to eventually decide, but their balancing can linger, like a tossed coin that lands on edge before eventually falling one way or the other \cite{glitch}.   The statistics of these devices exhibit a {\it half-life of indecision}, meaning that the probability of lingering indecision decreases by half after every successive elapse of the half-life.  An example of measured half-life is described in \cite{05aop}.  In a footnote to his 1905 paper introducing special relativity, Einstein explicitly mentioned the gap separating his concept of time from experimental experience:
\begin{quote}
The inexactness that lurks in the concept of the
simultaneity of two events at (approximately) the same place
must be bridged by an abstraction that will not be discussed here  \cite{05einstein}.
[Our translation.]\end{quote} 
\section{Computational networks in living organisms}\label{sec:5}
It seems safe to assume that brains compute, but what else they do is less
clear.  The lack of clarity stems in part from ignorance about how brains
function, but that lack is compounded by vagueness in defining {\it
  computation}.  We have tried to clarify computation by defining computation as
representable by live and safe marked graphs.  On that basis we propose to view
cognition in all organisms, from the basal cognition of bacteria to human
thinking, as involving both computation and something unpredictable beyond
computation, represented by unpredictable changes in live and safe marked
graphs.

We contrast our modeling of computation by computational networks with the modeling of brains, a long-standing topic which has persisted from McCulloch and Pitts \cite{McPitts} to the present \cite{18neurCompRev}.  Rather than modeling brains, we model computation itself, whether done by electronics or by a living organism.  As meant here, the elements of computational networks by which we model computation are not anatomical elements but logical elements.  We avoid discussing how these logical elements relate to anatomy.

We suggest that living organisms exhibit multiple computational networks.
Examples include: (1) the logically connected biochemical reaction steps of the
citric acid cycle (and its reverse) \cite{lane}, (2) neuronal networks with their
interacting rhythms \cite{buzsaki}, (3) physiological networks that regulate
heart beat and breathing.  Most of these networks undergo changes, some of which
are unpredictable.  These include:
\begin{enumerate}
\item   Hebbian synchronization during neuronal development \cite{hebb};
\item  the evolution of species via transposable genetic elements, so that one species, instead of having to have independently evolved a certain computational network, can, so to speak, borrow it;
\item  changes via genetic shifts in response to environmental pressures.
\end{enumerate}

Where we say computation, many authors say ``information processing,'' with {\it information} referring to Shannon's ``Mathematical Theory of Communication'' \cite{shannon48}.  In that theory, information is conveyed over given {\it channels}, but the theory does not address the construction, maintenance, and dissolution of channels.  If we view a computational network as a network of channels, then computation consists in information flowing over channels, and the part of cognition beyond computation corresponds to the unpredictable construction, maintenance and dissolution of channels. As an example of channel re-arrangement,  the slime mold {\it Physarum polycephalum} conveys information via locally oscillating cytoplasm flowing through channels formed of actin and myosin  \cite{23slime}.  These channels  are in a continual flux of assembly and disassembly.  In the marked-graph  representation of computation as information processing, channels are expressed by arrows of live and safe marked graphs, and changes in channels are expressed by changes in graphs that re-arrange the arrows.

\subsection{Modes of operation in computational networks}\label{sec:modebio}
The modeling of computational networks by marked graphs subject to extra-logical revisions yields a technical result of likely biological relevance. A directed graph of logic gates and their interconnections describes the static  structure of a computational network, but more is needed to describe its function.  As discussed in \cref{sec:3.1,sec:3.5}, most marked graphs that represent computational networks  have multiple families, and these express multiple modes of operation, even when only  a single mode of operation is desired.  A physiological example is in figure  11.3 of {\it Rhythms of the Brain}, captioned ``Multiple excitatory  glutamatergic loops in the hippocampal formation and associated structures'' \cite [p.\ 286]{buzsaki}. The graph in that figure shows nodes in opposing  simple directed circuits with cyclic orders $(RE,3,S)$ and its opposite  $(RE,S,3)$.  By \Cref{prop:koh}, this implies multiple families and hence  multiple modes of operation.  An analogy is waking up and being momentarily disoriented: seeing things ``wrongly.'' While graphs without markings are used  widely in biology, only the distinct families possible for marked graphs distinguish among multiple modes of operation.  

\subsection{Logical distance in biological computing}\label{sec:5.2}
An additional technical result is the concept of logical distance, defined
\cref{sec:3.3}.  {\it Logical distance} is the number of firings of a node that
must occur for a label produced by the node of a computational network to
circulate to a distant node and come back to the sender.  We conjecture that
logical distance underpins the evolution of both animal navigation and the
concepts of time and distance in people.

\subsection{Logical synchronization in human conversation and thought}
I (JMM) like logical synchronization as a metaphor for ``connecting'' in human conversation.  Sometimes a friend and I are ``synced'': we hear each other.  At other times, no one hears what anyone else is saying.  In both digital technology and human conversation, ``clicking into synchrony'' so that communication works is unpredictable and cannot be hurried. I picture my thoughts as comprising  fragments of computational networks in a flux of unpredictable modifications and interactions, coming into and out of logical synchronization.  At unpredictable moments, fragments of networks of thought can meet and fuse with one another, possibly to fission later, responding sometimes to spontaneous inner life and at other times to external stimuli.

\section{Discussion}\label{sec:6}
The demonstration of essential guesswork draws theoretical attention to the experimental investigations of the distinct physical conditions on which computational networks depend.  For a prime example, logical synchronization of a computational network requires  frequent adjustment of clock rates, guided by guesswork that anticipates unpredictable deviations.   Contrary to a linear ``begin--end'' view of computation, computational networks operate in a cycle of guess, test, and debug.  By denying ``final answers,'' the proof of guesswork elevates debugging to the level of a necessity.  By  recognizing the unpredictability of changes, this cycle is an analogy for both the scientific approach and biological evolution.

We represent computational networks by graphs with live and safe families. Computations must {\it move} but the mathematically defined marked graphs expressed in formulas sit still on the page. Throughout physics, the gap between physical motion and mathematical stillness is bridged by interpretations beyond the expressive power of mathematics, as when a succession of snapshots evokes a sense of motion.  We call attention to the demonstration in \cref{sec:4} of a computational network that admits no sequence of unambiguous global snapshots but can be related to a marked graph by local snapshots.

Our formulation of computational networks is deliberately restricted to their
logical aspects, making computational networks and their extra-logical changes
suitable for characterizing cognition in living organisms. Changes in the
corresponding live and safe marked graphs are a promising topic for further
investigation \cite{80Murata,81Murata}.

Guesses are sometimes private to a person, but they occur in social contexts.  Of the guesses that shape my life, I (JMM) borrow many from other people, living or dead.  I associate with guesses a social role assigned by D. W. Winnicott to ``illusory experience'' as a ``natural root of grouping among human beings'' \cite[Chapter 1.]{Winn}.  From our proof of the necessity of guesswork, what is fairly called illusory is not the difference among guesses that shape different groups, but the assumption that any one guess has an absolute claim to being ``right.''  We are indebted to Kenneth Augustyn for the thought that what is real to each of us is partly shared but in some other and often the most vivid parts highly personal \cite{ken}, leading to a sometimes dramatic tension between the shared and the personal.

\section*{Acknowledgments}
We learned about marked graphs from Anatol W.\ Holt and Frederick Commoner, and
drew extensively on \cite{71MarkedGraph}.  We are indebted to the late
C.\ A.\ Petri for spurring us to think hard about clocks in relation to
logic. We are indebted to Samuel J. Lomonaco, Jr.\ for helpful discussions and, in particular, 
 for pointing us to multigraphs.  We thank Kenneth A. Augustyn for
substantial help in organizing this paper.  We also thank David Mumford, Marcus
Appleby, and Johann Summhammer for comments on an earlier draft.  We are
grateful to Gy\"orgy Buzs\'aki for explaining one of his graphs of neural
connections.  We thank Stefan Haar for introducing us to
partial cyclic orders and for insightful comments.
\appendix

\section{Appendix: Mathematics of live and safe marked graphs}\label{sec:G}
Most of what we have to say about marked-graph mathematics is borrowed from \cite{71MarkedGraph}, in which a {\it marked graph} is defined as a finite multigraph (without loops) along with markings related by a firing \vspace*{3pt}rule.
We augment such marked graphs by labeling tokens according to the labeling rules for  NAND and FORK, as in \cref{fig:1a,fig:1b} (labels on out-arrows of ENV are \vspace*{3pt} {\bf unpredictable}).\footnote{Marked graphs with labels are a specialization of colored Petri
nets \cite{96Petri}, except that we allow multigraphs with loops.}
In addition we make the following two changes.
\begin{enumerate}
\item We allow loops in multigraphs, and we claim that the proofs in \cite{71MarkedGraph} carry over to multigraphs with loops.
\item In emphasizing ``snapshots,'' we explicitly flag as non-trivial the issue of associating timeless mathematical propositions with  the physical motion of logical operations.
\end{enumerate}

\subsection{Glossary and properties of marked graphs}\label{sec:A1}
The definitions of terms for graphs vary among authors. We  adapt terms from \cite{sedgwick}.
\begin{defn}
  \label{defn:d2}
  A {\it graph} consists of a finite set of elements that we call {\it nodes} and a finite set of elements that we call {\it arrows}, along with a function that assigns to each arrow an ordered pair of nodes, (as illustrated in preceding figures). \footnote{The words {\it node} and {\it arrow} evoke experiential associations. Formally, one has the following definition: a graph consists of a pair of finite sets and a  function that assigns  an ordered pair of elements of the  second set to each element of the first set.  However, that is more formal than we are prepared to be.} In the jargon, our graphs are  ``finite, directed multigraphs
with \vspace*{4pt}loops.''
\end{defn}
\noindent We say that an arrow {\it connects} the first node of its pair to the second node.
The first node is the {\it source} of the arrow, and the second node is the
{\it target} of the arrow.  An arrow is an {\it out-arrow} of its source
node and an {\it in-arrow} of its target node.  The source and target nodes of
an arrow need not be distinct, in which case the arrow is a {\it loop}.
More than one arrow can connect a given pair of nodes.  \\
A {\it directed path} in a graph is a sequence of arrows joined head to tail at nodes, with no repeated arrows.\\
A {\it simple} directed path has no repeated nodes.\\
A {\it circuit} is a closed directed path (which can be a {\it loop})\\ 
  A {\it simple directed circuit} is  a closed simple directed path.\\
    A graph is {\it strong}
    if it contains a directed path from every node to every other node. \\
      An {\it induced subgraph} of a graph $G$ is any graph consisting of a subset of nodes of $G$ and all the arrows of $G$ that connect pairs of those nodes. \\
      We call a simple  path from a node $v_1$ to a node $v_2$ a {\it 1-1 path} if (a) it consists of a single arrow or (b) its nodes other than its endpoints each have one in-arrow and one out-arrow.\\
      We call a node with one in-arrow and one out-arrow a
      {\it 1-1 node}.
\begin{defn}
An {\it unlabeled marking} of a graph is an assignment of a (non-negative)
  whole number, called the {\it token number}, to each arrow. In this report,
  that number is either 0 or 1.  A {\it labeled marking} accompanies an
  assignment of 1 (representing token presence) with a label consisting of a string of 0s
  and 1s (sometimes just a single 0 or 1).
\end{defn}
\noindent A node is {\it fireable} in a marking if the token number on each in-arrow of the node \vspace*{3pt}
is at least 1.\\  
For an unlabeled marked graph, a {\it firing} of a node is a relation between two markings: markings $M$ and $M'$ are related by the firing of node $v$ if the token number of $M'$ for each in-arrow of node $v$ 1 lower than the corresponding token number for $M$, and if the token number  $M'$ for each out-arrow of node $v$ is  1 higher than the corresponding token number for \vspace*{3pt}$M$.\\
A marking is called {\it live} if every node is either fireable, or can be made so through some sequence of firings.
\noindent Theorem 12 of \cite{71MarkedGraph} states the following proposition.
\begin{prop}
   If a live marking M' of a strong graph can produce M (through a sequence of firings), then M can produce M'. 
\end{prop}
\noindent Hence, ``live markings of a strongly connected graph partition into equivalence classes. \ldots
Let us refer to each equivalence class as a
{\it family}'' \cite[Thm. 12]{71MarkedGraph}. It follows that any marking of a family can be produced from any other by some sequence of \vspace*{3pt}firings.\\
We call the sum of token numbers over each arrow of a circuit  the {\it token count of the circuit}.  Because the firing of a node belonging to a circuit balances the decrease in the token number on the in-arrow of the node with the increase in the token number on the circuit out-arrow, we have the following\vspace*{4 pt} proposition.
\begin{prop}\label{prop:const}
  The token count of a circuit is the same for markings related by firings, and hence the same for all markings in any family \vspace*{4 pt} {\normalfont \cite{71MarkedGraph}}. 
\end{prop}

A generalization of this proposition is key to the definition of  logical distance in \cref{sec:3.3}.     
As exemplified for the nodes $v_0$ and $v_5$ of \cref{fig:nCy}, the union of a directed path with a return path need not be a directed circuit, but the pair of paths still has a property of preserving the sum of their token counts.
\begin{prop}\label{prop:handy}
For any  marked graph with a path, say $P$, connecting two nodes and a return path $Q$ connecting the same two nodes in the opposite direction, the number of tokens on $P$ plus the number of tokens on  $Q$  is invariant under node
  firings and is thus the same for all markings of any family.
\end{prop}
\begin{proof}
The  number of tokens on a directed path changes by +1 when the initial node fires, by $-1$ when the terminal node fires, and is invariant under the firing of its other nodes.  The terminal node of the directed path from $v$ to $w$ is the initial node of the path from $w$ to $v$ and vice versa, so that the sum of tokens on the two paths is invariant under all node firings.  
\end{proof}

To use marked graphs to represent interconnected logical operations, one must  restrict the graphs and the markings to avoid the piling up of more than one token on any arrow.  This means requiring safe markings.
\begin{defn}
  A marking is called {\it safe} if no arrow has a token number greater than 1 and if no sequence of firings can lead to a token number greater than 1 on any arrow.
\end{defn}
\noindent {\bf To represent computing mathematically, we invoke strong graphs with live and safe families.}  Live and safe markings possible only for strong graphs.\footnote{Among graphs that are connected in the weak sense that ignores arrow directions, a graph can have a live and safe marking if and only if it is strongly connected.}

\begin{defn}\label{defn:m1}
 If, for a marking $M$, two or more nodes are fireable, then we say those nodes are {\em concurrently fireable} in $M$.\footnote{(1) Defining concurrency for marked graphs is simpler than for Petri nets.\\(2) The token game corresponding to a marked graph with a family of live and safe markings can be played solitaire to generate a repeating cycle of firings so that each node of the graph fires once  in each cycle.  For graphs on which concurrent firings are possible, there are multiple such cycles of firings.  This corresponds to a partial cyclic order as defined by Stehr \cite{98stehr}. Stehr's definition does not always distinguish between  cyclic orders corresponding to distinct families, nor does the somewhat different definition provided by Haar \cite{cyclic}.
 }
\end{defn}
\subsection{Fragments of graphs for addition, illustrating graph contractions}\label{sec:frag}
Numerical comparisons (as in the thermostat of \cref{fig:therm1}) involve
addition.  \Cref{fig:full} shows addition performed by gates\cite{add} and 
illustrates coarse descriptions using contracted fragments of 
 \vspace{-12 pt} graphs.  \captionsetup[figure]{font=small,skip=0pt}
\begin{figure}[H]\hspace*{.5 in}
 \includegraphics[height=3 in]{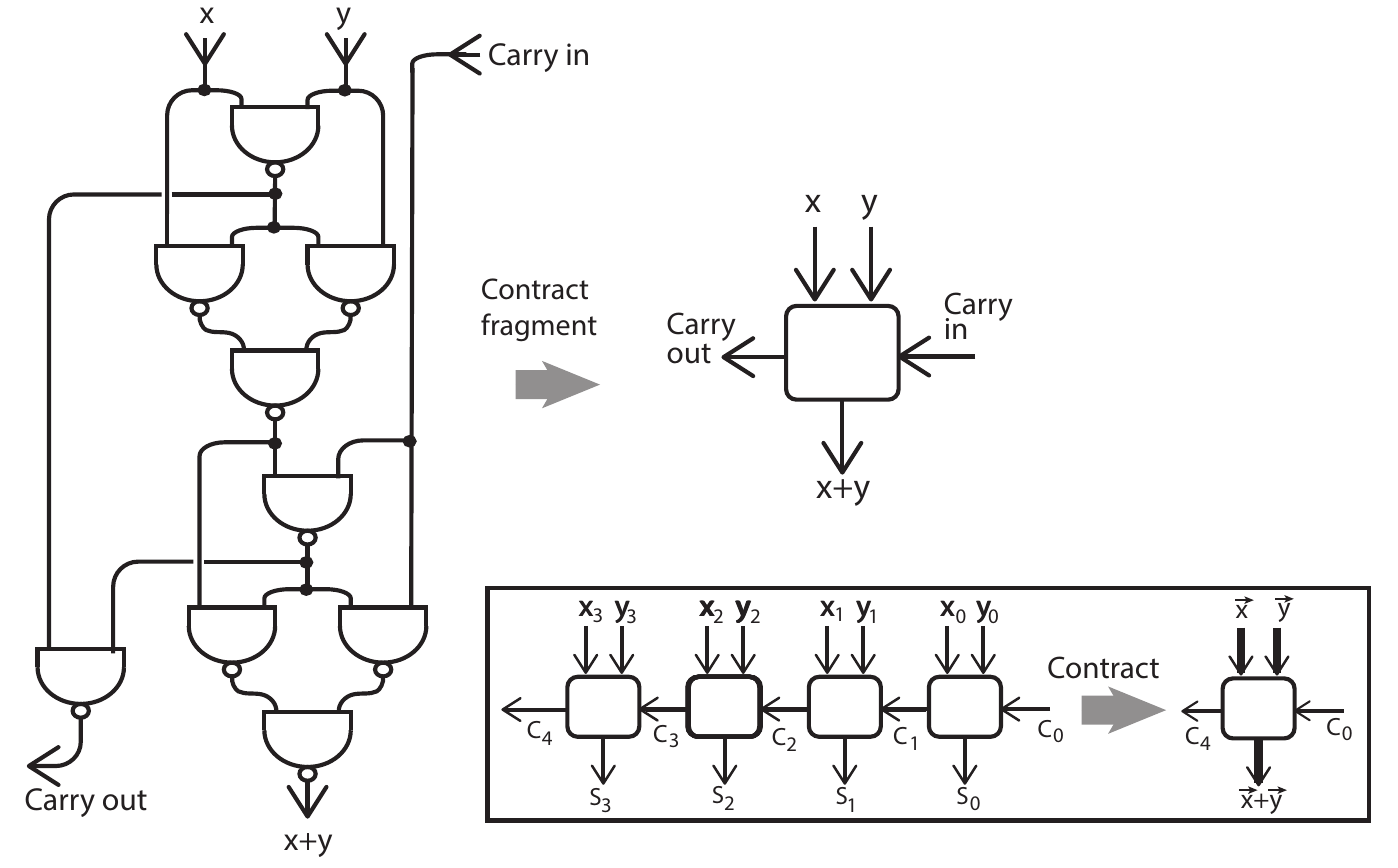}
 \caption{A graph fragment for 1-bit full adder with contraction to a node at a
   coarser level of description. Panel: the joining of 1-bit adders into a
   larger fragment to make a 4-bit adder, followed by contraction. $\stackrel{\rightarrow}{x}$ and
$\stackrel{\rightarrow}{y}$ denote 4-bit strings.\label{fig:full} }
\end{figure}

\subsection{Program control in marked graphs}\label{sec:progControl}

While our purpose is not to express general-purpose computers, this can be done
by marked graphs that have paths of token propagation under program control.
For instance, \cref{fig:cross}(a) pictures a contracted fragment of a graph
representing a switching element controlled by input {\bf c}.  If {\bf c} is 0,
then the token label of $X$ goes to $A$ and the token label of $Y$ goes to $B$, while
if {\bf c} is 1, then the token label of $X$ goes to $B$ and the token label of $Y$
goes to $A$.  \Cref{fig:cross}(b) shows the switching element as an uncontracted
fragment of a \vspace*{-12pt} graph.
\captionsetup[figure]{font=small,skip=0pt}
\begin{figure}[H]\hspace*{1 in}
 \includegraphics[height=2.7 in]{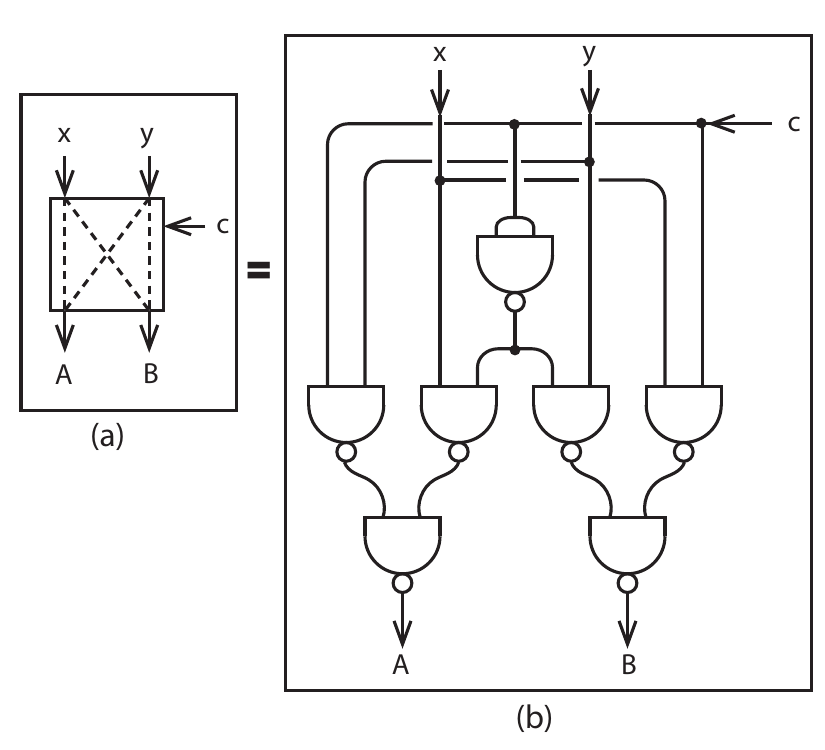}
 \caption{(a) Contracted switching element; (b) switching element as a fragment of a graph.\label{fig:cross}
 }
\end{figure}
\Cref{fig:fswap} shows how the input ${\bf c}$ from the environment controls the routing of other inputs {\bf x} and {\bf y} through a computation.  \begin{figure}[H]\hspace*{1 in}
 \includegraphics[height=2.8 in]{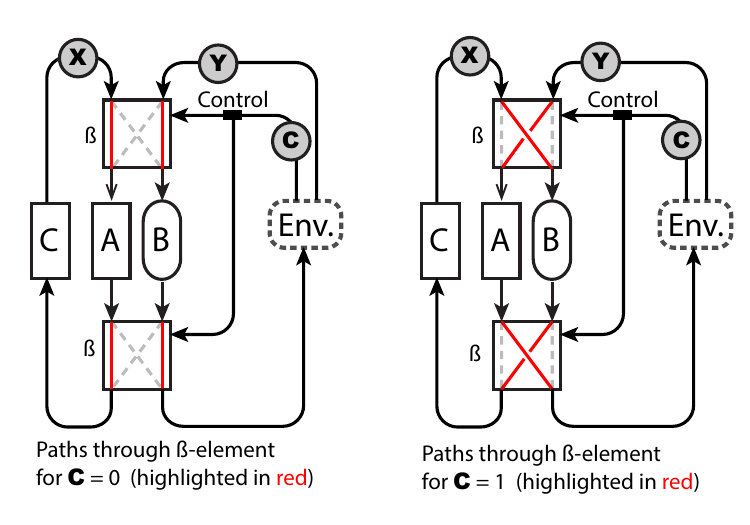}
 \caption{Example of environment controlling computational structure.\label{fig:fswap}
 }
\end{figure}

\subsection{Counting through variation in token labels}
Asymptotically,  all nodes fire equally frequently, but token labels allow
the expression of rhythms of bit values.  For example, \cref{fig:1to3}
shows a marked graph in which  node $V_7$ produces a 1 bit every third cycle
and a 0 bit on the two cycles in between.
\begin{figure}[H]
 \includegraphics[height=2.5 in]{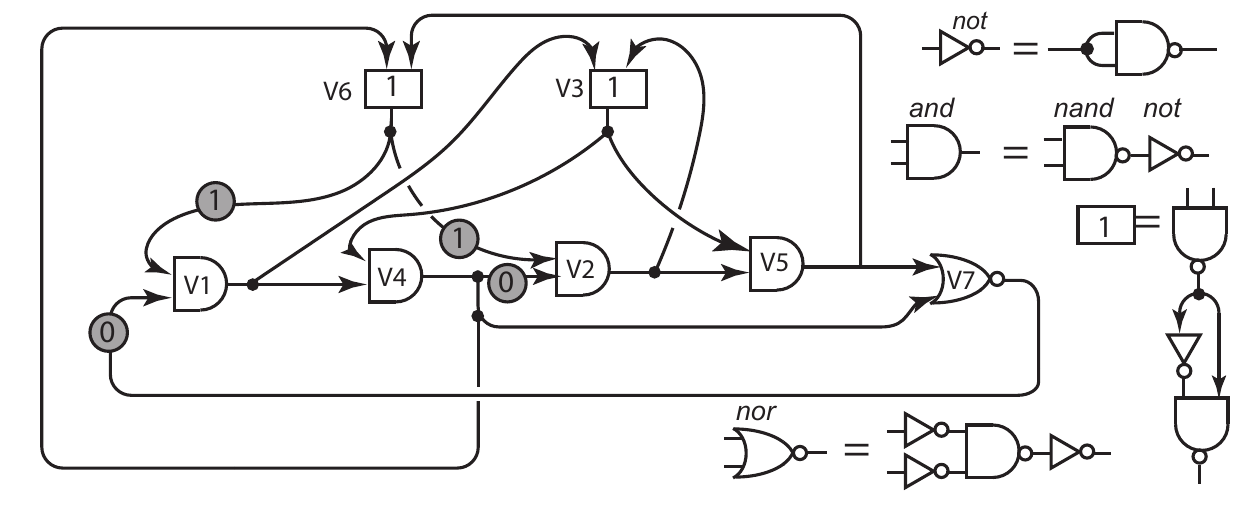}
 \caption{\label{fig:1to3}Marked graph that divides a rhythm by 3.
 }
\end{figure}

\end{document}